\documentclass{amsart}

\usepackage{amsmath}
\usepackage{amssymb,amsthm}
\usepackage{mathtools}
\usepackage{pgfplots}
\pgfplotsset{compat=1.16}
\usepackage{gastex}

\usepackage{cite}

\usepackage[boxed]{algorithm}
\usepackage[noend]{algorithmic}
  \newlength\commentspace
  \newcommand\algcomment[2]{\makebox[0pt][l]{\hspace{-#1em}%
    \hspace{\commentspace}$\triangleright$ #2}}
\newcommand{\algorithmicfunc}[1]{\textsc{#1}}
\newcommand{\FUNC}[1]{\item[\algorithmicfunc{#1}]}

\usepackage{graphicx}
\usepackage{array}
\newcolumntype{L}[1]{>{\raggedright\let\newline\\\arraybackslash\hspace{0pt}}m{#1}}
\newcolumntype{C}[1]{>{\centering\let\newline\\\arraybackslash\hspace{0pt}}m{#1}}
\newcolumntype{R}[1]{>{\raggedleft\let\newline\\\arraybackslash\hspace{0pt}}m{#1}}

\DeclareMathOperator{\Sing}{\mathrm{Sing}}

\DeclareRobustCommand{\stirling}{\genfrac\{\}{0pt}{}}

\usepackage{float}
\usepackage{subfig}

\usepackage{multirow}
\usepackage{rotating}

\usepackage{mathptmx}      
\theoremstyle{plain}
\newtheorem{lemma}{Lemma}
\newtheorem{theorem}{Theorem}
\newtheorem{corollary}{Corollary}
\newtheorem{proposition}{Proposition}

\theoremstyle{remark}
\newtheorem{examplenew}{Example}
\newtheorem{remarknew}{Remark}

\DeclareSymbolFont{rsfscript}{OMS}{rsfs}{m}{n}
\DeclareSymbolFontAlphabet{\mathrsfs}{rsfscript}

\DeclareMathOperator{\lf}{\mathrm{leaf}}
\DeclareMathOperator{\ch}{\mathrm{child}}

\newcommand{\dt}{.}
\DeclareMathOperator{\excl}{\mathrm{excl}}
\DeclareMathOperator{\dupl}{\mathrm{dupl}}

\newcommand{\sa}{synchronizing automata}
\newcommand{\san}{synchronizing automaton}
\newcommand{\cra}{completely reachable automata}
\newcommand{\cran}{completely reachable automaton}

\newcommand{\rl}{reset threshold}
\newcommand{\rt}{reset threshold}

\newcommand{\scn}{strongly connected}

\newcommand{\mA}{\mathrsfs{A}}

\newcommand{\mC}{\mathrsfs{C}}
\newcommand{\mE}{\mathrsfs{E}}
\newcommand{\ov}{\overline}

\usepackage{enumerate}
\usepackage{url}
\begin{document}

\title[Completely Reachable Automata]{Completely reachable automata:\\ an interplay between automata, graphs, and trees}

\author[E. A. Bondar, D. Casas, M. V. Volkov]{Evgeniya A. Bondar \and David Casas \and Mikhail V. Volkov}

\address{Institute of Natural Sciences and Mathematics, Ural Federal University\\ 620000 Ekaterinburg, Russia\\
bondareug@gmail.com, dafecato4@gmail.com, m.v.volkov@urfu.ru}

\thanks{The authors were supported by the Ministry of Science and Higher Education of the Russian Federation, project FEUZ-2020-0016.}

\begin{abstract}
A deterministic finite automaton in which every non-empty set of states occurs as the image of the whole state set under the action of a suitable input word is called completely reachable. We characterize such automata in terms of graphs and trees.
\end{abstract}

\maketitle

\section{Introduction }
\label{sec:intro}

\emph{Completely reachable automata} are complete deterministic finite automata in which every non-empty subset of the state set occurs as the image of the whole state set under the action of a suitable input word. This class of automata appeared in the study of descriptional complexity of formal languages~\cite{Maslennikova:2012} and in relation to the \v{C}ern\'{y} conjecture~\cite{Don16}. Recently, some species of \cra\ have been studied in~\cite{Gonzeetal2019,Maslennikova:2019,Hoffmann21,Hoffmann21a}. Two of the present authors have addressed \cra\ in the conference papers~\cite{BondarVolkov16,BondarVolkov18}. Here we strengthen and/or correct several results from these papers; besides, the article adds some essentially new material.

The article is structured as follows. In Section~\ref{sec:background} we first recall the necessary definitions and introduce notation; then we relate \cra\ to widely studied \sa. In Section~\ref{sec:characterization} we characterize \cra\ in terms of certain directed graphs; this characterization is more succinct than that presented in~\cite{BondarVolkov18}. Algorithmic issues of complete reachability are discussed in Section~\ref{sec:algortihm}. In~Section~\ref{sec:applications} we fix and prove a result announced in~\cite{BondarVolkov18} and refute a conjecture from~\cite{BondarVolkov16}. In Section~\ref{sec:resetthreshold} we study synchronization of \cra. We prove the observation from~\cite{BondarVolkov16} that the \v{C}ern\'{y} conjecture holds for \cra\ with two input letters and give some partial results for \cra\ with arbitrary input alphabets. Section~\ref{sec:final} contains a discussion of some directions for further research.

\section{Background and Motivation}
\label{sec:background}

\subsection{Graphs}
\label{subsec:condensation}
We start by specifying the variant of graph-theoretic terminology utilized in this article. Here and below, we use expressions like $A:=B$ to emphasize that $A$ is defined to be $B$.

A \emph{graph} is a quadruple $\Gamma:=\langle V,E,s,t\rangle$ of sets $V,E$ and maps $s,t\colon E\to V$. Elements of $V$ and $E$ are called \emph{vertices} and, respectively, \emph{edges} of the graph. For each edge $e\in E$, the vertices $s(e)$ and $t(e)$ are called the \emph{source} and, respectively, the \emph{target} of $e$. Notice that different edges may share the same source and target; such edges are called \emph{parallel}. (Thus, our graphs are in fact directed multigraphs, sometimes called \emph{quivers} in the literature.) Graphs without parallel edges are called \emph{simple}. In a simple graph, every edge $e$ is uniquely determined by the pair $(s(e),t(e))$, and we take the liberty of identifying $E$ with the subset $\{(s(e),t(e))\mid e\in E\}$ of $V\times V$. Hence, a simple graph may be thought of as a pair $\langle V,E\rangle$, where $E\subseteq V\times V$.

Two edges $e,e'$ of a graph $\Gamma=\langle V,E,s,t\rangle$ are \emph{consecutive} if $t(e)=s(e')$. A \emph{path} in $\Gamma$ is a finite sequence $e_1,\dots,e_\ell$ of edges such that for each $i=1,\dots,\ell-1$, the edges $e_i,e_{i+1}$ are consecutive; the number $\ell$ is called the \emph{length} of the path. The empty sequence is also treated as a path (of length 0) referred to as the \emph{empty path}. We say that a path \emph{starts} or \emph{originates} at the source of its first edge and \emph{ends} or \emph{terminates} at the target of its last edge. We adopt the convention that the empty path may start at any vertex. A vertex $v'$ is said to be \emph{reachable} from a vertex $v$ in $\Gamma$ if $\Gamma$ has a path that starts at $v$ and ends at $v'$.

The \emph{reachability} relation on $V$ is the set of all pairs $(p,q)\in V\times V$ such that $q$ is reachable from $p$. Clearly, the reachability relation is a pre-order, that is, a reflexive (thanks to the above convention about the empty path) and transitive relation. The equivalence corresponding to the pre-order is called the \emph{mutual reachability} relation. Its classes are called \emph{clusters} of $\Gamma$. A graph with a unique cluster is said to be \emph{strongly connected}.

The \emph{condensation} of a graph $\Gamma=\langle V,E,s,t\rangle$ is a simple graph, denoted $\Gamma^{\mathsf{con}}$, whose vertices are the clusters of $\Gamma$. A pair $(C,C')$ of different clusters is an edge in $\Gamma^{\mathsf{con}}$ if and only if there is at least one edge $e\in E$ with $s(e)\in C$ and $t(e)\in C'$; the edge $(C,C')$ is said to be \emph{induced} by $e$.

We need the notions of a tree and a forest. Under the framework presented above, we define a \emph{tree} as a simple graph $\langle V,E\rangle$ that has a vertex $r$, called \emph{root}, such that for every vertex $v\in V$, there exists a unique path that starts at $r$ and ends at $v$. If $(v,u)$ is an edge in a tree, $u$ is said to be a \emph{child} of $v$ and $v$ is said to be the \emph{parent} of $u$. Observe that our definition of a tree (combined with our convention about the empty path) implies that the root has no parent while each other vertex has exactly one parent. The set of all children of a vertex $v$ is denoted $\ch(v)$. A vertex that has no children is called a \emph{leaf}.

A \emph{forest} is a disjoint union of trees. If $\Phi=\langle V,E\rangle$ is a forest and $v\in V$, consider the set of vertices in $V$ that are reachable from $v$. The subgraph of $\Phi$ induced on this set is a tree whose root is $v$. The tree is called the \emph{subtree rooted at $v$} and is denoted by $\Phi_v$. The set of all leaves of the tree $\Phi_v$ is called the \emph{leafage} of $v$ and denoted $\lf(v)$.

\subsection{Automata}
\label{subsec:automata}
In this paper, a \emph{complete deterministic finite automaton} (DFA) is a triple
\begin{equation}
\label{eq:automaton}
\mathrsfs{A}:=\langle Q,\Sigma,\delta\rangle.
\end{equation}
Here $Q$ and $\Sigma$ are finite sets called the \emph{state set} and, respectively, the \emph{input alphabet} of $\mA$, and $\delta\colon Q\times\Sigma\to Q$ is a totally defined map called the \emph{transition function} of $\mA$.

The elements of $\Sigma$ are referred to as \emph{letters} and finite sequences of letters are called \emph{words over $\Sigma$}. The empty sequence is also treated as a word over $\Sigma$ called the \emph{empty word} and denoted $\varepsilon$. The collection of all words over $\Sigma$ denoted $\Sigma^*$ forms a monoid under the operation of word concatenation.

The transition function $\delta$ extends to a function $Q\times\Sigma^*\to Q$ (still denoted $\delta$) via the following recursion: for every $q\in Q$, we set $\delta(q,\varepsilon):=q$  and $\delta(q,wa):=\delta(\delta(q,w),a)$ for all $w\in\Sigma^*$ and $a\in\Sigma$. Thus, every word $w\in\Sigma^*$ induces the transformation $q\mapsto\delta(q,w)$ of the set $Q$.

Let $\mathcal{P}(Q)$ stand for the set of all non-empty subsets of the set $Q$. The function $\delta$ can be further extended to a function $\mathcal{P}(Q)\times\Sigma^*\to\mathcal{P}(Q)$ (again denoted by $\delta$) by letting $\delta(P,w):=\{\delta(q,w)\mid q\in P\}$ for every non-empty subset $P\subseteq Q$.

Whenever we deal with a fixed DFA, we tend to simplify our notation by suppressing the sign of the transition function; this means that we write $q\dt w$ for $\delta(q,w)$ and $P\dt w$ for $\delta(P,w)$ and introduce a DFA as a pair $\langle Q,\Sigma\rangle$.

It is convenient (in particular, in illustrations) to represent a DFA $\mathrsfs{A}=\langle Q,\Sigma\rangle$ by a labeled graph with the vertex set $Q$, in which for all $q,q'\in Q$ and $a\in\Sigma$ such that $q\dt a=q'$, we have an edge with source $q$, target $q'$ and label $a$, denoted $q\xrightarrow{a}q'$. We refer to this labeled graph as to the \emph{underlying graph} of $\mA$. A DFA is called \scn\ is so is its underlying graph.

\subsection{Completely reachable automata and synchronization}
\label{subsec:preliminaries}

Given a DFA $\mathrsfs{A}=\langle Q,\Sigma\rangle$, we say that a non-empty subset $P\subseteq Q$ is \emph{reachable} in $\mathrsfs{A}$ if $P=Q\dt w$ for some word $w\in\Sigma^*$. A DFA is called \emph{completely reachable} if every non-empty subset of its state set is reachable. A DFA $\mathrsfs{A}=\langle Q,\Sigma\rangle$ is called \emph{synchronizing} if it has a reachable singleton, that is, $Q\dt w$ is a singleton for some word $w\in\Sigma^*$. Any such word $w$ is said to be a \emph{reset word} for the DFA. Comparing the definitions, we immediately arrive at the following observation.

\begin{lemma}
\label{lem:first} Every \cran\ is synchronizing.
\end{lemma}

Synchronizing automata serve as transparent and useful models of error-resistant systems in many applied areas (system and protocol testing, information coding, robotics). At the same time, \sa\ surprisingly arise in some parts of pure mathematics (symbolic dynamics, theory of non-negative matrices, substitution systems, convex optimization, theory of permutation groups, and others). Basic aspects of the theory of \sa\ as well as some of its diverse connections and applications are discussed, for instance, in the surveys~\cite{Sandberg:2005,Jurgensen:2008,Volkov:2008,Bogdanovic&Imreh&Ciric&Petkovic:1999,Mateescu&Salomaa:1999,Ananichev&Petrov&Volkov:2006,Cherubini:2007,ArCaSt17}, in Chapters~3 and~10 of the monograph~\cite{berstel2009codes}, and in the chapter~\cite{KV} of the ``Handbook of Automata Theory''. Of course, the importance of \sa\ does not yet justify that their subclass consisting of \cra\ also is of interest, but here we provide some evidence making us believe that the subclass is indeed worth attention.

The \emph{length} $|w|$ of a word $w\in\Sigma^*$ is defined in the usual way: we set $|w|:=0$ if $w$ is empty and $|w|:=|v|+1$ if $w=va$ for some word $v\in\Sigma^*$ and some letter $a\in\Sigma$. The minimum length of reset words for a \san\ $\mathrsfs{A}$ is called the \emph{\rl} of $\mathrsfs{A}$.

One of the central issues in the theory of \sa\ is the question of how the \rt\ of a \san\ depends on the number of states. In~1964 \v{C}ern\'{y}~\cite{Cerny1964} constructed for each $n>2$ a \san\ $\mathrsfs{C}_n$ with $n$ states, two input letters, and \rl\ $(n-1)^2$. If we denote the states of $\mathrsfs{C}_n$ by $0,1,\dots,n-1$ and the input letters by $a$ and $b$, the actions of the letters are as follows:
\begin{equation}
\label{eq:cerny}
i\dt a:=\begin{cases}
i &\text{if }\ i>0,\\
1 &\text{if }\ i=0;
\end{cases}
\qquad
i\dt b:=i+1\!\!\pmod{n}.
\end{equation}
The automaton $\mathrsfs{C}_n$ is shown in Fig.\,\ref{fig:cerny-n}. Here and below we adopt the convention that edges bearing multiple labels represent two or more parallel edges. In particular, the edge $0\xrightarrow{a,b}1$ in Fig.\,\ref{fig:cerny-n} represents the two parallel edges $0\xrightarrow{a}1$ and $0\xrightarrow{b}1$.
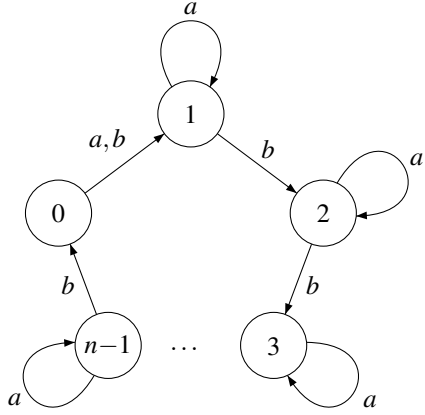
\begin{figure}[bht]
\begin{center}
\unitlength .55mm
\begin{picture}(72,90)(0,-85)
\gasset{Nw=16,Nh=16,Nmr=8}
\node(n0)(36.0,-16.0){1}
\node(n1)(4.0,-40.0){0} \node(n2)(68.0,-40.0){2}
\node(n3)(16.0,-72.0){$n{-}1$} \node(n4)(56.0,-72.0){3}
\drawedge[ELdist=2.0](n1,n0){$a,b$} \drawedge[ELdist=1.5](n2,n4){$b$}
\drawedge[ELdist=1.7](n0,n2){$b$}
\drawedge[ELdist=1.7](n3,n1){$b$}
\drawloop[ELdist=1.5,loopangle=30](n2){$a$}
\drawloop[ELdist=2.4,loopangle=-30](n4){$a$}
\drawloop[ELdist=1.5](n0){$a$}
\drawloop[ELdist=1.5,loopangle=210](n3){$a$}
\put(31,-73){$\dots$}
\end{picture}
\caption{The automaton $\mathrsfs{C}_n$}\label{fig:cerny-n}
\end{center}
\end{figure}

The DFAs $\mathrsfs{C}_n$ are well-known in the connection with the famous \emph{\v{C}ern\'{y} conjecture} about the maximum \rl\ for \sa\ with $n$ states: the series $\{\mathrsfs{C}_n\}$ provides the lower bound $(n-1)^2$ for this maximum, and the conjecture claims that every \san\ with $n$ states can be reset by a word of length $(n-1)^2$. Besides representing the worst possible synchronization behaviour, the automata $\mathrsfs{C}_n$ possess other interesting properties, including the following one:

\begin{examplenew}
\label{examp:cerny}
Each automaton $\mathrsfs{C}_n$, $n>1$, is completely reachable.
\end{examplenew}
Example~\ref{examp:cerny} was explicitly registered by Maslennikova in the proof of~\cite[Proposition~2]{Maslennikova:2012}. In an implicit form, however, this example is contained in a much earlier result due to McAlister~\cite[Theorem~4.3]{McAlister:1998} who provided a comprehensive analysis of the subsemigroup generated by the transformations $a$ and $b$ from \eqref{eq:cerny} in the monoid of all transformations of the set $\{0,1,\dots,n-1\}$.

Recall that the \v{C}ern\'{y} conjecture, first stated in the 1960s, resists researchers' efforts for more than 50 years. The best so far upper bound for the \rl\ of \sa\ with $n$ states is cubic in $n$; it is due to Shitov~\cite{Shitov:2019} who has slightly improved the bound established by  Szyku\l{}a~\cite{Szykula:2018}. In turn, Szyku\l{}a's bound is only slightly better than the upper bound $\frac{n^3-n}6$ established by Pin~\cite{Pin:1983} and Frankl~\cite{Frankl:1982}, and independently by Klyachko, Rystsov, and Spivak~\cite{KlyachkoRystsovSpivak87}, more than 30 years ago.  Since the conjecture has proved to be very hard in general, a pragmatic approach taken by many authors is to study it under certain additional restrictions,
see, e.g., \cite{Ry97,Dubuc98,Eppstein:1990,Kari03,AV2004,Ananichev:2005,Trahtman:2007,Volkov:2009,STEINBERG20115487,Steinberg:2011,GrechK13a}. In fact, the conjecture has been shown to hold in various special classes of DFAs,  and the class of \cra\ appears to constitute quite a natural candidate within such an approach. Indeed, the class consists entirely of \sa\ by Lemma~\ref{lem:first}, and in view of Example~\ref{examp:cerny}, it is rich enough to contain automata that are conjectured to be extremal with respect to the \rt.

\section{Characterization in Terms of Graphs}
\label{sec:characterization}
In this section we characterize completely reachable automata in terms of strong connectivity of certain simple graphs. Our characterization is similar in spirit to~\cite[Theorem 5]{BondarVolkov18}, but our present construction is more economic and, as we hope, more transparent than the construction used in~\cite{BondarVolkov18}.

\subsection{The graph $\Gamma_1(\mathrsfs{A})$}
\label{subsec:gamma1}

As in~\cite{BondarVolkov18}, our present characterization of \cra\ relies on an iterative process that assigns to each given DFA $\mathrsfs{A}=\langle Q,\Sigma\rangle$ a certain simple graph. The process starts with the graph  $\Gamma_1(\mathrsfs{A})$ defined in~\cite{BondarVolkov16}. For the reader's convenience, we reproduce the definition here.

Given a DFA $\mathrsfs{A}=\langle Q,\Sigma\rangle$, the \emph{defect} of a word $w\in\Sigma^*$ with respect to $\mathrsfs{A}$ is defined as the cardinality of the set difference $Q{\setminus}Q\dt w$. If a word $w$ has defect~1, the difference $Q{\setminus}Q\dt w$ consists of a unique state, which is called the \emph{excluded state} for $w$ and denoted by $\excl(w)$. Further, the set $Q\dt w$ contains a unique state $p$ such that $p=q_1\dt w=q_2\dt w$ for some $q_1\ne q_2$; this state $p$ is called the \emph{duplicate state} for $w$ and denoted by $\dupl(w)$. See Fig.~\ref{fig:exdu}, in which for a typical word $w$ of defect 1, the map $q\mapsto q\dt w$ is visualized as a bipartite graph.

\begin{figure}[th]
\begin{center}
\unitlength .9mm
\begin{picture}(80,45)(0,-5)
\gasset{Nw=2,Nh=2,Nmr=1}
\node(u0)(0,38){}
\node(u1)(15,38){}
\node(u2)(30,38){}
\put(37,38){$\dots$}
\node(u3)(50,38){}
\node(u4)(65,38){}
\node(u5)(80,38){}
\node(b0)(0,2){}
\node(b1)(15,2){}
\node(b2)(30,2){}
\put(37,2){$\dots$}
\node(b3)(50,2){}
\node(b4)(65,2){}
\node(b5)(80,2){}
\put(-8,38){$Q$}
\put(-8,0){$Q$}
\put(-8,20){$w$}
\drawedge(u0,b3){}
\drawedge(u1,b0){}
\drawedge(u2,b5){}
\drawedge(u3,b3){}
\drawedge(u4,b2){}
\drawedge(u5,b4){}
\put(10,-5){$\excl(w)$}
\put(45,-5){$\dupl(w)$}
\end{picture}
\caption{Excluded and duplicate states of a word}\label{fig:exdu}
\end{center}
\end{figure}
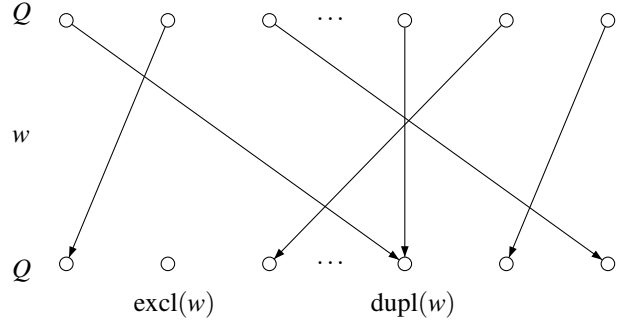

Let $W_1(\mathrsfs{A})$ stand for the set of all words of defect~1 with respect to $\mathrsfs{A}$. Define the graph $\Gamma_1(\mathrsfs{A}):=\langle Q_1,E_1\rangle$ with the vertex set $Q_1:=Q$ and the edge set
\begin{equation}
\label{eq:e1}
E_1:=\{(\excl(w),\dupl(w))\mid w\in W_1(\mathrsfs{A})\}.
\end{equation}
We say that the edge $(\excl(w),\dupl(w))$ is \emph{forced} by the word $w$.

Our running example, which will illustrate stages of our construction, is based on the automaton $\mathrsfs{E}_5$ with the state set $\{1,2,3,4,5\}$ and eight input letters, denoted $a_{[1]}$, $a_{[2]}$, $a_{[3]}$, $a_{[4]}$, $a_{[5]}$, $a_{[1,2]}$, $a_{[4,5]}$, and $a_{[1,3]}$. The transition function of $\mathrsfs{E}_5$ is given in Table~\ref{tb:e5}.
\begin{table}[ht]
\caption{The transition table of the automaton $\mathrsfs{E}_5$}\label{tb:e5}
\begin{center}
\begin{tabular}{c@{\ }|@{\ }c@{\ }c@{\ }c@{\ }c@{\ }c@{\ }c@{\ }c@{\ }c}
            $q$  & $q\dt a_{[1]}$ & $q\dt a_{[2]}$ & $q\dt a_{[3]}$  & $q\dt a_{[4]}$ & $q\dt a_{[5]}$ & $q\dt a_{[1,2]}$ & $q\dt a_{[4,5]}$ & $q\dt a_{[1,3]}$ \mathstrut\\
\hline
1 & 2 & 1 & 1 & 1 & 1 & 3 & 1 & 4\mathstrut\\
2 & 2 & 1 & 1 & 2 & 2 & 3 & 1 & 4\mathstrut\\
3 & 3 & 3 & 2 & 3 & 3 & 3 & 2 & 4\mathstrut\\
4 & 4 & 4 & 4 & 5 & 4 & 4 & 3 & 5\mathstrut\\
5 & 5 & 5 & 5 & 5 & 4 & 5 & 3 & 5\mathstrut
\end{tabular}
\end{center}
\end{table}
(It is the same DFA that appeared in~\cite[Example~2]{BondarVolkov18}, where a typo crept in: the transition table in~\cite[Example~2]{BondarVolkov18} gave $5\dt a_{[2]}=4$ while it should have been $5\dt a_{[2]}=5$ as shown in Table~\ref{tb:e5}.) The automaton $\mathrsfs{E}_5$ is shown in Fig.~\ref{fig:e5} where loops are omitted to improve readability.

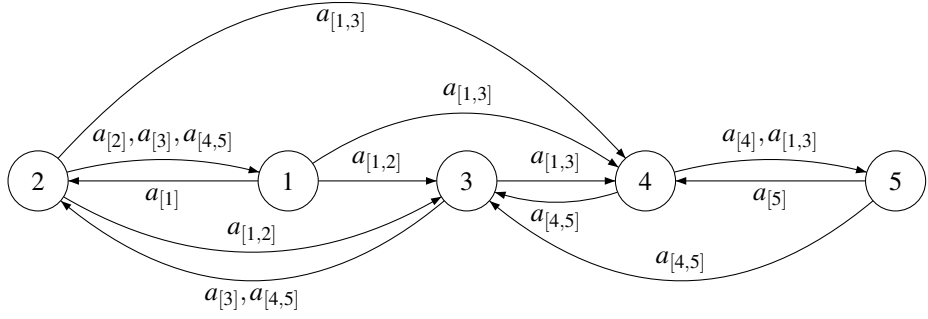
\begin{figure}[ht]
\begin{center}
\unitlength=0.95mm
\begin{picture}(115,45)(-10,-20)
\node(A1)(-10,0){2}
\node(B1)(25,0){1}
\node(C1)(50,0){3}
\node(D1)(75,0){4}
\node(E1)(110,0){5}
\drawedge[curvedepth=3](A1,B1){$a_{[2]},a_{[3]},a_{[4,5]}$}
\drawedge(B1,A1){$a_{[1]}$}
\drawedge[curvedepth=3](D1,E1){$a_{[4]},a_{[1,3]}$}
\drawedge(E1,D1){$a_{[5]}$}
\drawedge(B1,C1){$a_{[1,2]}$}
\drawedge[curvedepth=-10](A1,C1){$a_{[1,2]}$}
\drawedge[curvedepth=14](C1,A1){$a_{[3]},a_{[4,5]}$}
\drawedge(C1,D1){$a_{[1,3]}$}
\drawedge[curvedepth=10,eyo=-1](B1,D1){$a_{[1,3]}$}
\drawedge[ELside=r,curvedepth=25](A1,D1){$a_{[1,3]}$}
\drawedge[curvedepth=3](D1,C1){$a_{[4,5]}$}
\drawedge[ELside=r,curvedepth=14](E1,C1){$a_{[4,5]}$}
\end{picture}
\caption{The automaton $\mathrsfs{E}_5$ with loops omitted}\label{fig:e5}
\end{center}
\end{figure}

In order to construct the edges of the graph $\Gamma_1(\mathrsfs{E}_5)$ according to \eqref{eq:e1}, one has to compute the transformations caused by words of defect~1 with respect to $\mathrsfs{E}_5$. Clearly, any word of defect~1 in $\mathrsfs{E}_5$ must be a product of letters of defect~1, that is, a product of some of the letters $a_{[1]}$, $a_{[2]}$, $a_{[3]}$, $a_{[4]}$, $a_{[5]}$. Multiplying all pairs of these letters, one gets the following `multiplication table', in which $*$ means that the defect of the product is larger than 1:
\begin{center}
\begin{tabular}{ c| c | c | c | c | c |}
          & $a_{[1]}$ & $a_{[2]}$ & $a_{[3]}$ & $a_{[4]}$ & $a_{[5]}$ \\
\hline
$a_{[1]}$ & $a_{[1]}$ & $a_{[2]}$  & $a_{[3]}$ & $*$ & $*$\\
\hline
$a_{[2]}$ & $a_{[1]}$ & $a_{[2]}$  & $a_{[3]}$ & $*$ & $*$\\
\hline
$a_{[3]}$ & $*$ & $*$  & $*$ & $*$ & $*$\\
\hline
$a_{[4]}$ & $*$ & $*$  & $*$ & $a_{[4]}$ & $a_{[5]}$\\
\hline
$a_{[5]}$ & $*$ & $*$  & $*$ & $a_{[4]}$ & $a_{[5]}$\\
\hline
\end{tabular}
\end{center}
Hence, every word in $W_1(\mathrsfs{E}_5)$ causes the same transformation as one of the letters $a_{[1]}$, $a_{[2]}$, $a_{[3]}$, $a_{[4]}$, $a_{[5]}$. Therefore the graph $\Gamma_1(\mathrsfs{E}_5)$ has five edges forced by the five letters. Table~\ref{tb:e5} shows that $\excl(a_{[i]})=i$, whence the letter $a_{[i]}$ forces the edge with the source $i$. The graph is depicted in Fig.~\ref{fig:g1e5}.
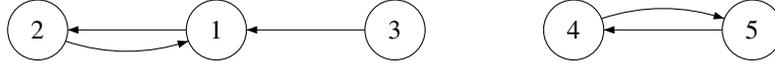
\begin{figure}[hbt]
\begin{center}
\unitlength=0.95mm
\begin{picture}(100,15)(0,-10)
\node(A1)(0,0){2}
\node(B1)(25,0){1}
\node(C1)(50,0){3}
\node(D1)(75,0){4}
\node(E1)(100,0){5}
\drawedge[curvedepth=-3](A1,B1){}
\drawedge(B1,A1){}
\drawedge(C1,B1){}
\drawedge[curvedepth=3](D1,E1){}
\drawedge(E1,D1){}
\end{picture}
\caption{The graph $\Gamma_1(\mathrsfs{E}_5)$}\label{fig:g1e5}
\end{center}
\end{figure}

In this particular example, analyzing all transformations caused by words of defect 1 turns out to be easy. In general, an exhaustive analysis of all such transformations is not feasible because for an automaton with $n$ states, their number may reach $\binom{n}2n!$. Nevertheless, Gonze and Jungers \cite{GJ:2019} have managed to develop an algorithm that, given a DFA $\mathrsfs{A}$, builds the graph $\Gamma_1(\mathrsfs{A})$ in polynomial time of the size of the automaton.

\subsection{The graphs $\Gamma_k(\mathrsfs{A})$ for $k>1$}
\label{subsec:gammak}

For the next steps of our construction, we extend the operators $\excl(\_)$ and $\dupl(\_)$ to arbitrary words. Namely, given a DFA $\mathrsfs{A}=\langle Q,\Sigma\rangle$ and a word $w\in\Sigma^*$, we let
\[
\excl(w):=Q{\setminus}Q\dt w\ \text{ and } \dupl(w):=\{p\in Q\mid p=q_1\dt w=q_2\dt w \ \text{ for some }\ q_1\ne q_2\}.
\]
If we take the usual liberty of ignoring the distinction between singleton sets and their elements, then for words of defect~1, the new meanings of $\excl(w)$ and $\dupl(w)$ agree with the definition in~Subsection~\ref{subsec:gamma1}.

Given a DFA $\mathrsfs{A}=\langle Q,\Sigma\rangle$, we aim to assign a graph $\Gamma(\mathrsfs{A})$ to it. As mentioned, our process starts with the graph $\Gamma_1:=\Gamma_1(\mathrsfs{A})=\langle Q_1,E_1\rangle$. If $\Gamma_1$ is \scn, then $\Gamma(\mathrsfs{A}):=\Gamma_1$ and the process stops with SUCCESS. If all clusters of $\Gamma_1$ are singletons, we also set $\Gamma(\mathrsfs{A}):=\Gamma_1$ and the process stops with FAILURE.

Beyond these two extreme cases, our present construction deviates from that of~\cite[Section 3]{BondarVolkov18}. Now we proceed, postponing a comparison between the two constructions till Subsection~\ref{subsec:comparison}.

Thus, suppose that the graph $\Gamma_1$ is not \scn\ and possesses a cluster of size at least~2. Consider the condensation $\Gamma_1^{\mathsf{con}}$ of $\Gamma_1$. Let $Q_2$ stand for the vertex set of $\Gamma_1^{\mathsf{con}}$; recall that $Q_2$ is the set of clusters of $\Gamma_1$ so that vertices of $\Gamma_1^{\mathsf{con}}$ are subsets of $Q_1$. We denote the edge set of $\Gamma_1^{\mathsf{con}}$ by $\overline{E}_1$, meaning that every edge in $\overline{E}_1$ is induced by an edge from $E_1$ as explained in Subsection~\ref{subsec:condensation}.

Let $W_2(\mathrsfs{A})$ denote the set of all words of defect~2 with respect to $\mathrsfs{A}$. We let $E_{\leqslant2}:=\overline{E}_1\cup E_2$, where
\begin{equation}
\label{eq:e2}
\begin{split}
E_2:=\{(C,C')\in Q_2\times Q_2\mid C\ne C', \ C\supseteq\excl(w),\\ C'\cap\dupl(w)\ne\varnothing\ \text{for some}\ w\in W_2(\mathrsfs{A})\}.
\end{split}
\end{equation}
Extending the terminology used for $E_1$, we say that the edge $(C,C')\in E_2$ with $C\supseteq\excl(w)$, $C'\cap\dupl(w)\ne\varnothing$ is \emph{forced} by the word $w$. Now we define $\Gamma_2(\mathrsfs{A})$ as the graph $\langle Q_2,E_{\leqslant2}\rangle$. Thus, the graph $\Gamma_2(\mathrsfs{A})$ is $\Gamma_1^{\mathsf{con}}$ to which we append all edges forced by words of defect~2.

\smallskip

For an illustration, Fig.~\ref{fig:g2e5} shows the graph $\Gamma_2(\mathrsfs{E}_5)$.
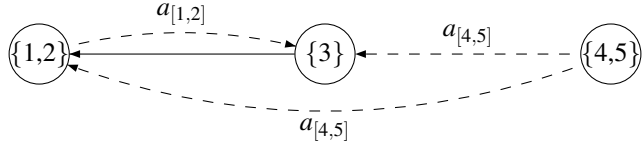
\begin{figure}[htb]
\begin{center}
\unitlength=0.95mm
\begin{picture}(100,18)(0,-5)
\node(C2)(50,5){\{3\}}
\node(F2)(10,5){\{1,2\}}
\node(G2)(90,5){\{4,5\}}
\drawedge(C2,F2){}
\drawedge[dash={1.5}{1.5},curvedepth=3](F2,C2){$a_{[1,2]}$}
\drawedge[ELside=r,dash={1.5}{1.5}](G2,C2){$a_{[4,5]}$}
\drawedge[curvedepth=8,dash={1.5}{1.5}](G2,F2){$a_{[4,5]}$}
\end{picture}
\caption{The graph $\Gamma_2(\mathrsfs{E}_5)$}\label{fig:g2e5}
\end{center}
\end{figure}
The solid edge is the only one `inherited' from $\Gamma_1(\mathrsfs{E}_5)$ in the course of condensation. (In passing from a graph $\Gamma$ to its condensation $\Gamma^{\mathsf{con}}$ only edges between states in different clusters of $\Gamma$ matter.) The edges from $E_2$ are dashed and labeled by words of defect 2 that force them. Observe that  $a_{[4,5]}$ forces two edges at once---the definition \eqref{eq:e2} permits this. Of course, one has to check that no word $w\in W_2(\mathrsfs{E}_5)$ may force the `missing' edge in Fig.~\ref{fig:g2e5}, that is, $\bigl(\{1,2\},\{4,5\}\bigr)$. This amounts to verifying that for each word $w$ whose image is $\{3,4,5\}$, the only state in $\dupl(w)$ is 3. We omit the verification as it is absolutely routine.

\smallskip

We continue with our inductive definition of $\Gamma(\mathrsfs{A})$. Now we aim to define the graph $\Gamma_{k}(\mathrsfs{A})$ for $k>2$ provided that the graph $\Gamma_{k-1}(\mathrsfs{A})$ has already been defined. The overall plan follows the above pass from $\Gamma_{1}(\mathrsfs{A})$ to $\Gamma_{2}(\mathrsfs{A})$: if the graph $\Gamma_{k-1}(\mathrsfs{A})$ is not \scn\ but possesses `sufficiently large' clusters, we first build its condensation $\Gamma_{k-1}^{\mathsf{con}}$ and then append edges forced by words of defect $k$ with respect to $\mA$. A subtlety that we have to address is a correct definition of being `sufficiently large'.

Let $Q_{k-1}$ and $E_{\leqslant k-1}$ stand for the vertex and the edge sets of the graph $\Gamma_{k-1}(\mathrsfs{A})$. According to the plan outlined, the vertex set $Q_k$ of $\Gamma_{k}(\mathrsfs{A})$ is defined as the set of clusters of $\Gamma_{k-1}(\mA)$. In turn, $Q_{k-1}$ is the set of clusters of $\Gamma_{k-2}(\mA)$, etc. This defines a forest structure on the set
$$Q_1\cup Q_2\cup\dots\cup Q_{k-1}\cup Q_k;$$
namely, for each $i=1,\dots,k-1$, we consider $C\in Q_i$ as a child of $D\in Q_{i+1}$ if and only if $C\in D$. We refer to the forest obtained this way as the \emph{forest of clusters} of the sequence $\Gamma_1(\mA),\Gamma_2(\mA),\dots,\Gamma_k(\mA)$ and denote this forest by $\mathcal{F}_k(\mA)$.

\smallskip

For an illustration, Fig.~\ref{fig:f3e5} shows the forest $\mathcal{F}_3(\mathrsfs{E}_5)$ where $\mathrsfs{E}_5$ is our running example. Here we represent edges of trees by segments rather than arrows, following the standard convention of edges of trees always going downwards. For clarity, we have maintained the labels of vertices in $Q_2$ and $Q_3$ even though this information is redundant as the labels can be readily recovered by climbing up the trees of $\mathcal{F}_3(\mathrsfs{E}_5)$.
\begin{figure}[thb]
\begin{center}
\unitlength 0.65mm
\begin{picture}(90,71)(0,3)
\gasset{AHnb=0}
\node(1)(10,10){$1$}
\node(2)(30,10){$2$}
\node(3)(50,10){$3$}
\node(6)(70,10){$4$}
\node(7)(90,10){$5$}
\gasset{Nadjust=wh}
\node(4)(20,40){$\{1,2\}$}
  \drawedge(1,4){}
  \drawedge(2,4){}
\node(5)(30,70){$\{\{1,2\},\{3\}\}$}
  \drawedge(4,5){}
\node(8)(80,40){$\{4,5\}$}
  \drawedge(6,8){}
  \drawedge(7,8){}
\node(9)(50,40){$\{3\}$}
  \drawedge(3,9){}
  \drawedge(9,5){}
\node(10)(80,70){$\{\{4,5\}\}$}
  \drawedge(8,10){}
\put(-10,8){$Q_1$}
\put(-10,38){$Q_2$}
\put(-10,68){$Q_3$}
\end{picture}
\caption{The forest $\mathcal{F}_3(\mathrsfs{E}_5)$}\label{fig:f3e5}
\end{center}
\end{figure}
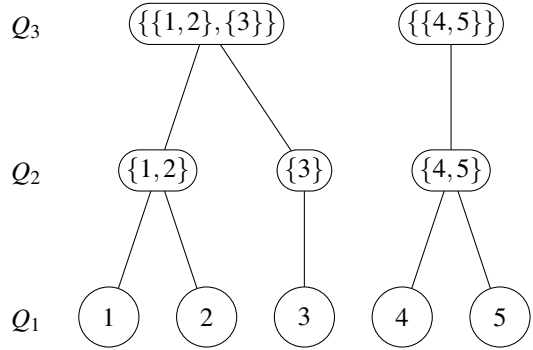

\begin{remarknew} One may ``compress'' the forest $\mathcal{F}_k(\mA)$ by identifying singleton sets with their elements, but such a compressed representation does not seem to give any real advantage.
\end{remarknew}

Now we are in a position to complete the inductive definition of the graph $\Gamma(\mathrsfs{A})$. If the graph $\Gamma_{k-1}(\mathrsfs{A})$ is \scn, then $\Gamma(\mathrsfs{A}):=\Gamma_{k-1}(\mathrsfs{A})$ and the process stops with SUCCESS. If the leafage of each vertex in $Q_k$ has less than $k$ elements, we also set $\Gamma(\mathrsfs{A}):=\Gamma_{k-1}(\mathrsfs{A})$ and stop with FAILURE. (Recall that the leafage of a vertex $C\in Q_k$, denoted $\lf(C)$, is the set of all leaves in the subtree of $\mathcal{F}_k(\mA)$ rooted at $C$.)

Otherwise, consider the condensation $\Gamma_{k-1}^{\mathsf{con}}$ and denote its edge set by $\overline{E}_{\leqslant k-1}$. Let $W_k(\mathrsfs{A})$ stand for the set of all words of defect~$k$ with respect to $\mathrsfs{A}$. We let $E_{\leqslant k}:=\overline{E}_{\leqslant k-1}\cup E_k$, where
\begin{equation}
\label{eq:ek}
\begin{split}
E_k:=\{(C,C')\in Q_k\times Q_k\mid C\ne C',\ \lf(C)\supseteq\excl(w),\\ \lf(C')\cap\dupl(w)\ne\varnothing\ \text{for some}\ w\in W_k(\mathrsfs{A})\}.
\end{split}
\end{equation}
Observe that, even though we have assumed that $k>2$, the definition \eqref{eq:ek} makes sense for $k=2$ and leads to exactly the same result as the definition of $E_2$ in \eqref{eq:e2} because $C=\lf(C)$ for any $C\in Q_2$. (Indeed, the leaves of a cluster $C\in Q_2$ are precisely the elements of $C$.) Preserving the terminology used for $E_1$ and $E_2$, we say that the edge $(C,C')\in E_k$ such that $\lf(C)\supseteq\excl(w)$ and $\lf(C')\cap\dupl(w)\ne\varnothing$ is \emph{forced} by the word $w$.

Finally, we define $\Gamma_k(\mathrsfs{A})$ as the graph $\langle Q_k,E_{\leqslant k}\rangle$.

\smallskip

Fig.~\ref{fig:g3e5} shows the graph $\Gamma_3(\mathrsfs{E}_5)$.
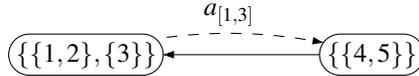
\begin{figure}[bht]
\begin{center}
\unitlength=0.95mm
\begin{picture}(80,12)(0,0)
\gasset{Nadjust=wh}
\node(F2)(20,5){$\{\{1,2\},\{3\}\}$}
\node(G2)(60,5){$\{\{4,5\}\}$}
\drawedge[dash={1.5}{1.5},curvedepth=3](F2,G2){$a_{[1,3]}$}
\drawedge(G2,F2){}
\end{picture}
\caption{The graph $\Gamma_3(\mathrsfs{E}_5)$}\label{fig:g3e5}
\end{center}
\end{figure}
The solid edge is `inherited' from $\Gamma_2(\mathrsfs{E}_5)$ in the course of condensation. The dashed edge comes from $E_3$ and is forced by $a_{[1,3]}$ since $\excl(a_{[1,3]})=\{1,2,3\}=\lf\left(\{\{1,2\},\{3\}\}\right)$ while $\dupl(a_{[1,3]})=\{4,5\}=\lf\left(\{\{4,5\}\}\right)$.

Observe that the graph $\Gamma_3(\mathrsfs{E}_5)$ is \scn. Hence $\Gamma(\mathrsfs{E}_5)=\Gamma_3(\mathrsfs{E}_5)$ and for the automaton $\mathrsfs{E}_5$, our process stops at step 3 with SUCCESS.

The construction of the graph $\Gamma_k(\mathrsfs{A})$ uses words of defect $k$ with respect to $\mA$. For a DFA with $n$ states, the maximum defect of a word is $n-1$, and therefore, the described process stops after at most $n-1$ steps. This obvious bound is tight: in Section~\ref{sec:applications} we provide examples of DFAs with $n$ states for which $n-1$ steps are needed to reach either SUCCESS or FAILURE.

\subsection{Characterization theorems}
\label{subsec:main}

Here we prove two results that reveal how the graph $\Gamma(\mathrsfs{A})$ helps in recognizing complete reachability of $\mA$. Theorem~\ref{thm:sufficient3} covers the case when constructing the graph terminates with SUCCESS while Theorem~\ref{thm:necessary} handles the case of FAILURE. The two theorems look the same as \cite[Theorem~3]{BondarVolkov18} and respectively \cite[Theorem~4]{BondarVolkov18}; we stress, nevertheless, that  Theorems~\ref{thm:sufficient3} and~\ref{thm:necessary} are essentially new since our present construction of the `indicator' graph $\Gamma(\mathrsfs{A})$ differs from the one used in \cite{BondarVolkov18}.

\begin{theorem}
\label{thm:sufficient3}
If a DFA $\mathrsfs{A}=\langle Q,\Sigma\rangle$ is such that the graph $\Gamma(\mathrsfs{A})$ is \scn\ and $\Gamma(\mathrsfs{A})=\Gamma_k(\mathrsfs{A})$, then $\mathrsfs{A}$ is completely reachable; more precisely, for every non-empty subset $P\subseteq Q$, there is a product $w$ of words of defect at most $k$ such that $P=Q\dt w$.
\end{theorem}

\begin{proof}
Take any non-empty subset $P\subseteq Q$. We prove that $P$ is reachable in $\mathrsfs{A}$ via a product of words of defect at most $k$ by induction on $|Q{\setminus} P|$. If $P=Q$, there is nothing to prove as $Q$ is reachable via the empty word. Now let $P$ be a proper subset of $Q$. We aim to find a subset $R\subseteq Q$ such that $P=R\dt w$ for some word $w$ of defect at most~$k$ and $|R|>|P|$. Since $|Q{\setminus} R|<|Q{\setminus} P|$, the induction assumption applies to the subset $R$ whence $R=Q\dt v$ for some product $v$ of words of defect at most $k$. Then $P=Q\dt vw$ is reachable as required.

\smallskip

Thus, fix a non-empty proper subset $P\subset Q$. Let $P_1:=P$ and for each $i=2,\dots,k$, let $P_i:=\{C\in Q_i\mid \lf(C)\subseteq P\}$. Then $P_i$ is a proper subset of $Q_i$ for each~$i$ since $P$ is a proper subset of $Q$.
We say that an edge $e$ of the graph $\Gamma_i(\mathrsfs{A})$ \emph{penetrates} $P_i$ if $s(e)\notin{P}_i$ while $t(e)\in {P}_i$. For some $i$, penetrating edges may not exist but such an edge certainly exists for $i=k$ because the graph $\Gamma_k(\mathrsfs{A})$ is \scn. Now let $m$ be the least number such that there is an edge in $\Gamma_m(\mathrsfs{A})$ that penetrates $P_m$.

First consider the case when $m=1$ and so $P_m=P$. Then there is some $e\in E_1$ such that $s(e)\notin P$ and $t(e)\in P$. Here the reasoning from the proof of \cite[Theorem~1]{BondarVolkov16} applies; we reproduce it for the reader's convenience. By the definition of $E_1$ (see \eqref{eq:e1}), there is a word $w$ of defect~1 with respect to $\mathrsfs{A}$ for which $s(e)$ is the excluded state and $t(e)$ is the duplicate state. Then $t(e)=q_1\dt w=q_2\dt w$ for some $q_1\ne q_2$. Since the excluded state $s(e)$ for $w$ does not belong to $P$, for each state $r\in P{\setminus}\{t(e)\}$, there exists a state $r'\in Q$ such that $r'\dt w=r$. Now letting $R:=\{q_1,q_2\}\cup\bigl\{r'\mid r\in P{\setminus}\{t(e)\}\bigr\}$, we conclude that $P=R\dt w$ and $|R|=|P|+1$; see Fig.~\ref{fig:case1} for an illustration.
\begin{figure}[htb]
\begin{center}
\unitlength .99mm
\begin{picture}(60,40)(-1,0)
\drawoval(5,20,16,30,10)
\drawoval(50,20,20,35,12)
\gasset{Nw=2,Nh=2,Nmr=1}
\node[Nw=0,Nh=0,Nmr=0](u0)(5,35){}
\node(u1)(-13,35){}
\put(-25,38){$s(e)=\excl(w)$}
\node[Nw=0,Nh=0,Nmr=0](u3)(50,37.5){}
\node(u4)(49,28){}
\put(51,28){$q_1$}
\node(u5)(52,16){}
\put(54,16){$q_2$}
\node[Nw=0,Nh=0,Nmr=0](b0)(5,5){}
\node(b1)(5,28){}
\put(1,23){$t(e)={}$}
\put(-1,19){$\dupl(w)$}
\node[Nw=0,Nh=0,Nmr=0](b3)(50,2.5){}
\put(62,20){$R$}
\put(-8,20){$P$}
\drawedge[ELside=r,eyo=.1](u3,u0){$w$}
\drawedge[eyo=-.1](b3,b0){$w$}
\drawedge[ELpos=45,linewidth=0.5](u1,b1){$e$}
\drawedge[eyo=.5](u4,b1){}
\drawedge[eyo=-.5](u5,b1){}
\node(r0)(6,10){}
\put(3,10){$r$}
\node(r1)(51,8){}
\put(53,8){$r'$}
\drawedge(r1,r0){}
\end{picture}
\end{center}
\caption{Proof of Theorem~\ref{thm:sufficient3}, case $m=1$}\label{fig:case1}
\end{figure}
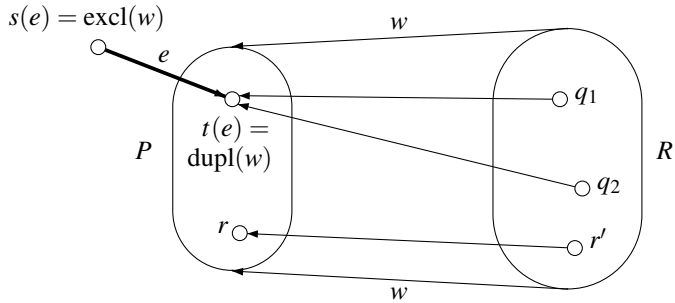

Now let $m>1$. We start with registering the following fact.

\begin{lemma}
\label{lem:alternative}
For each $D\in Q_{m'}$ with $m'\le m$, the set $\lf(D)$ is either contained in ${P}$ or disjoint from ${P}$.
\end{lemma}

\begin{proof} We induct on $m'$. If $m'=1$, each element in $Q_1=Q$ is just a state $q$, say, and either $q\in P$ or $q\notin P$. Since $\lf(q)=\{q\}$, either $\lf(q)\subseteq P$ or $\lf(q)\cap P=\varnothing$, as required.

Now let $1<m'\le m$. By the definition, any $D\in Q_{m'}$ is a cluster in $\Gamma_{m'-1}(\mathrsfs{A})$ and children of $D$ are elements of $Q_{m'-1}$. By the inductive assumption, for each $C\in\ch(D)$, the set $\lf(C)$ is either contained in ${P}$ or disjoint from ${P}$. Suppose, towards a contradiction, that some $D$ has both a child $C$ with $\lf(C)\cap P=\varnothing$ and a child $C'$ with $\lf(C')\subseteq P$. By the definition, $P_{m'-1}$ consists of elements of $Q_{m'-1}$ whose leafage is contained in $P$. Therefore, we have $C'\in P_{m'-1}$ while $C\notin P_{m'-1}$. Vertices in the cluster $D$ are reachable from each other. In particular, there must be a path $f_1,f_2,\dots,f_\ell$ in $\Gamma_{m'-1}(\mathrsfs{A})$ with $s(f_1)=C\notin P_{m'-1}$ and $t(f_\ell)=C'$. If $j$ is the maximal index with $s(f_j)\notin P_{m'-1}$, then either $j=\ell$ or $j<\ell$ and $s(f_{j+1})\in P_{m'-1}$. In either case we get $t(f_j)\in P_{m'-1}$. We see that the edge $f_j$ penetrates ${P}_{m'-1}$ so that the penetration occurs for $m'-1<m$. This contradicts our choice of $m$.

Thus, either $\lf(C)\subseteq P$ for all $C\in\ch(D)$ or $\lf(C)\cap P=\varnothing$ for all $C\in\ch(D)$. Since obviously $\lf(D)=\bigcup_{C\in\ch(D)}\lf(C)$, we conclude that $\lf(D)\subseteq P$ in the former case and $\lf(D)\cap P=\varnothing$ in the latter one.
\end{proof}

Fix an edge $e$ in $\Gamma_m(\mathrsfs{A})$ that penetrates $P_m$. The edge set $E_{\leqslant m}$ of the graph $\Gamma_m(\mathrsfs{A})$ has been defined as the union $\overline{E}_{\leqslant m-1}\cup E_m$ where edges in $\overline{E}_{\leqslant m-1}$ are induced by the edges of the graph $\Gamma_{m-1}(\mathrsfs{A})$ in the course of condensation while edges in $E_m$ are forced by words of defect $m$ with respect to $\mA$. We aim to show that the edge $e$ cannot belong to $\overline{E}_{\leqslant m-1}$.

Let $D:=s(e)$ and $D':=t(e)$. Recall that vertices of the graph $\Gamma_m(\mathrsfs{A})$ are clusters in $\Gamma_{m-1}(\mathrsfs{A})$; in particular, so are $D$ and $D'$. Now suppose that some edge $(C,C')\in E_{\leqslant m-1}$ induces the edge $e$; this means that $C\in D$ and $C'\in D'$. In terms of the forest $\mathcal{F}_m(\mA)$, the two memberships say that $C$ and $C'$ are children of $D$ and $D'$, respectively.  Since $D'$ belongs to the set $P_m$, the definition of $P_m$ implies $\lf(D')\subseteq P$. As $C'$ is a child of $D'$, we have $\lf(C')\subseteq\lf(D')$ whence $\lf(C')\subseteq P$, and therefore, $C'\in P_{m-1}$. On the other hand, $D\notin P_m$, whence $\lf(D)\nsubseteq P$. Then Lemma~\ref{lem:alternative} implies that $\lf(D)\cap P=\varnothing$ whence $\lf(C)\cap P=\varnothing$ and $C\notin P_{m-1}$. We see that the edge $(C,C')$ penetrates ${P}_{m-1}$, and this
contradicts our choice of $m$.

Thus, $e\in E_m$. By the definition (see \eqref{eq:ek}), there exists a word $w$ of defect~$m\le k$ with respect to $\mathrsfs{A}$ such that $\lf(D)\supseteq\excl(w)$ and $\lf(D')\cap\dupl(w)\ne\varnothing$. Now choose any state $p\in\lf(D')\cap\dupl(w)$. By the definition of $\dupl(w)$, there exist some $q_1,q_2\in Q$ such that $q_1\ne q_2$ and $q_1\dt w=q_2\dt w=p$. As observed in the preceding paragraph, we have $\lf(D)\cap P=\varnothing$, whence $\excl(w)\cap P=\varnothing$. Thus, for every state $r\in P{\setminus}\{p\}$, one can choose a state $r'\in Q$ such that $r'\dt w=r$. Hence, we can proceed as in case $m=1$: we set
$$R:=\{q\mid q\dt w=p\}\cup\bigl\{r'\mid r\in P{\setminus}\{p\}\bigr\}$$ and conclude that $P=R\dt w$ and $|R|>|P|$ since $q_1,q_2\in R$. The configuration used in this argument is illustrated in Fig.~\ref{fig:case2}.\hfill$\Box$
\end{proof}

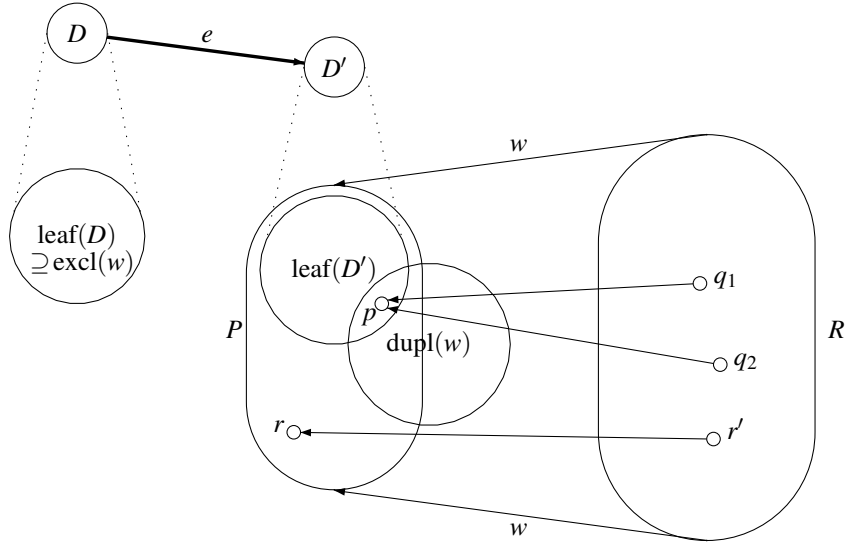
\begin{figure}[bht]
\begin{center}
\unitlength .9mm
\begin{picture}(100,80)(-20,-10)
\drawoval(15,20,26,45,15)
\drawoval(70,20,32,60,20)
\node(d0)(-23,65){$D$}
\node(d2)(15,60){$D'$}
\gasset{Nw=2,Nh=2,Nmr=1}
\node[Nw=0,Nh=0,Nmr=0](u0)(15,42.5){}
\rpnode[arcradius=10](u1)(-23,35)(30,10){$\lf(D)$}
\put(-30,30){${\supseteq}\excl(w)$}
\node[Nw=0,Nh=0,Nmr=0](u3)(70,50){}
\node(u4)(69,28){}
\put(71,28){$q_1$}
\node(u5)(72,16){}
\put(74,16){$q_2$}
\node[Nw=0,Nh=0,Nmr=0](b0)(15,-2.5){}
\rpnode[arcradius=10](b1)(15,30)(30,11){$\lf(D')$}
\rpnode[arcradius=12](d1)(29,19)(30,12){$\dupl(w)$}
\node[Nw=0,Nh=0,Nmr=0](b3)(70,-10){}
\node(p1)(22,25){}
\put(19,22.5){$p$}
\put(88,20){$R$}
\put(-1,20){$P$}
\drawedge[ELside=r,eyo=.1](u3,u0){$w$}
\drawedge[eyo=-.1](b3,b0){$w$}
\drawedge[linewidth=0.5](d0,d2){$e$}
\drawedge[eyo=.5](u4,p1){}
\drawedge[eyo=-.5](u5,p1){}
\node(r0)(9,6){}
\put(6,6){$r$}
\node(r1)(71,5){}
\put(73,5){$r'$}
\drawedge(r1,r0){}
\gasset{AHnb=0}
\drawedge[dash={0.2 1.2}0,sxo=-4.5,exo=-10](d0,u1){}
\drawedge[dash={0.2 1.2}0,sxo=4.5,exo=10](d0,u1){}
\drawedge[dash={0.2 1.2}0,sxo=-4.5,exo=-11](d2,b1){}
\drawedge[dash={0.2 1.2}0,sxo=4.5,exo=11](d2,b1){}
\end{picture}
\end{center}
\caption{Proof of Theorem~\ref{thm:sufficient3}, case $m>1$; dotted lines indicate trees rooted at $D$ and $D'$}\label{fig:case2}
\end{figure}

\begin{theorem}
\label{thm:necessary}
If a DFA $\mathrsfs{A}=\langle Q,\Sigma\rangle$ is such that the graph $\Gamma(\mathrsfs{A})$ is not \scn, then $\mathrsfs{A}$ is not completely reachable; more precisely, if $\Gamma(\mathrsfs{A})=\Gamma_k(\mathrsfs{A})$, then some subset in $Q$ with at least $|Q|-k$ states is not reachable in $\mathrsfs{A}$.
\end{theorem}

\begin{proof}
Assume that $\Gamma(\mathrsfs{A})=\Gamma_k(\mathrsfs{A})$ is not \scn. We denote the set of clusters of $\Gamma(\mathrsfs{A})$ by $Q_{k+1}$. The reachability relation on the vertex set of $\Gamma(\mathrsfs{A})$ induces a partial order $\preceq$ on $Q_{k+1}$. Fix a cluster $C$ which is minimal with respect to $\preceq$. Since the construction of the graph $\Gamma(\mathrsfs{A})$ has stopped, we have $|\lf(C)|\le k$. Let $\ell:=|\lf(C)|$ and $P:=Q{\setminus}\lf(C)$. Then $|P|=|Q|-\ell\ge|Q|-k$.

We aim to show that the subset $P$ is not reachable in $\mathrsfs{A}$. Towards a contradiction, suppose that $P=Q\dt w$ for some $w\in\Sigma^*$. Then we have $\excl(w)=Q{\setminus}P=\lf(C)$. The defect of $w$ is $\ell\le k$, whence edges forced by $w$ lie in the set $E_\ell$ that occurs on the $\ell$-th step in the construction of $\Gamma(\mathrsfs{A})$. Let $(D,D')$ be an edge of the graph $\Gamma_\ell(\mathrsfs{A})=\langle Q_\ell,E_{\leqslant\ell}\rangle$ forced by $w$. By the definition, $D$ and $D'$ are different clusters of $\Gamma_{\ell-1}(\mathrsfs{A})$; besides that, the requirement  $\lf(D)\supseteq\excl(w)=\lf(C)$ must hold. If $\lf(D)\supset\lf(C)$, then $D$ would be a predecessor of $C$ in the forest $\mathcal{F}_{k+1}(\mA)$ which is clearly impossible. This means that $\lf(D)=\lf(C)$ and $D$ is a descendant of $C$ in $\mathcal{F}_{k+1}(\mA)$. Let $D_0:=D$, $D_1\in Q_{\ell+1}$, \dots, $D_{k-\ell}\in Q_k$ , $D_{k-\ell+1}:=C\in Q_{k+1}$ be the sequence of clusters that one traverses when climbing from $D$ to $C$ in the subtree of $\mathcal{F}_{k+1}(\mA)$ rooted at $C$. Consider also a similar sequence starting at $D'_0:=D'\in Q_\ell$: for $i=0,1,\dots,k-\ell$, let $D'_{i+1}\in Q_{\ell+i+1}$ be the parent of $D'_i\in Q_{\ell+i}$ in $\mathcal{F}_{k+1}(\mA)$. Since $\lf(D_i)=\lf(C)$ for all $i=0,1,\dots,k-\ell$, the clusters $D_i$ and $D'_i$ are different for each $i=0,1,\dots,k-\ell+1$. Therefore, in the course of subsequent condensations, the edge $(D,D')=(D_0,D'_0)$ induces the edge $(D_1,D'_1)\in\overline{E}_{\leqslant \ell+1}$ that in turn induces the edge $(D_2,D'_2)\in\overline{E}_{\leqslant \ell+2}$, etc. Finally, we arrive at the edge $e:=(D_{k-\ell},D'_{k-\ell})$ in the graph $\Gamma_k(\mathrsfs{A})=\Gamma(\mathrsfs{A})$. Then $s(e)\in C$ and $t(e)\in C':=D'_{k-\ell+1}$, whence $C'\prec C$.  This contradicts the minimality of $C$ with respect to the reachability order $\preceq$. \end{proof}

Combining Theorems~\ref{thm:sufficient3} and~\ref{thm:necessary}, we readily arrive at our main result.

\begin{theorem}
\label{thm:necess&suffic}
A DFA $\mathrsfs{A}$ is completely reachable if and only if the graph $\Gamma(\mathrsfs{A})$ is \scn.
\end{theorem}

\subsection{A comparison with an alternative construction of the indicator graph}
\label{subsec:comparison}

As mentioned, for $k>1$, our definition of $\Gamma_k(\mathrsfs{A})$ in Subsection~\ref{subsec:gammak} differs from the definition given in~\cite[Section 3]{BondarVolkov18}. In this subsection, we discuss the difference in some detail.

We reproduce the main definitions from \cite{BondarVolkov18}; for clarity, we add the superscript $\star$ to all objects taken from there, so that, say, $\Gamma^\star_2(\mathrsfs{A})$ stands for what was denoted by $\Gamma_2(\mathrsfs{A})$ in~\cite{BondarVolkov18}.

Let a DFA $\mA=\langle Q,\Sigma\rangle$ be such that the graph $\Gamma_1:=\Gamma_1(\mathrsfs{A})=\langle Q,E_1\rangle$ is not \scn\ and not all clusters of $\Gamma_1$ are singletons. Let $Q_2^\star$ be the collection of all at least 2-element clusters of the graph $\Gamma_1$. The graph $\Gamma^\star_2(\mathrsfs{A})$ has $Q\cup Q^\star_2$ as its vertex set. The edge set of $\Gamma^\star_2(\mathrsfs{A})$ is the union of $E_1$ with the set
\[
E^\star_2:=\{(C,p)\in Q^\star_2\times Q\mid C\supseteq\excl(w),\, p\in\dupl(w)\ \text{ for some }\ w\in W_2(\mathrsfs{A})\}
\]
and the set $I^\star_2:=\{(q,C)\in Q\times Q^\star_2\mid q\in C\}$ of \emph{inclusion} edges that represent the containments between the elements of $Q$ and the clusters in $Q^\star_2$.

For an illustration, see Fig.\,\ref{fig:gstar2e5} that displays the graph $\Gamma^\star_2(\mathrsfs{E}_5)$, where $\mE_5$ is the DFA used above as our running example. The inclusion edges are shown with dashed arrows while solid arrows represent the edges from $E_1\cup E^\star_2$.
\begin{figure}[hbt]
\begin{center}
\unitlength=0.95mm
\begin{picture}(105,25)(0,60)
\node(A2)(1,65){2}
\node(B2)(25,65){1}
\node(C2)(49,65){3}
\node(D2)(76,65){4}
\node(E2)(100,65){5}
\drawedge[curvedepth=-3](A2,B2){}
\drawedge(B2,A2){}
\drawedge(C2,B2){}
\drawedge[curvedepth=3](D2,E2){}
\drawedge(E2,D2){}
\node(F2)(13,82){\{1,2\}}
\drawedge[dash={1.5}{1.5}](A2,F2){}
\drawedge[dash={1.5}{1.5}](B2,F2){}
\drawedge[curvedepth=3](F2,C2){}
\node(G2)(88,82){\{4,5\}}
\drawedge[dash={1.5}{1.5}](D2,G2){}
\drawedge[dash={1.5}{1.5}](E2,G2){}
\drawedge[curvedepth=-3](G2,C2){}
\drawedge[curvedepth=-4](G2,B2){}
\end{picture}
\caption{The graph $\Gamma^\star_2(\mathrsfs{E}_5)$}\label{fig:gstar2e5}
\end{center}
\end{figure}

Now suppose that $k>2$ and the graph $\Gamma^\star_{k-1}(\mathrsfs{A})$  with the vertex set $Q\cup Q^\star_2\cup\cdots\cup Q^\star_{k-1}$ and the edge set
\begin{equation}
\label{eq:e3}
E_1\cup E^\star_2\cup\cdots\cup E^\star_{k-1}\cup I^\star_2\cup\cdots\cup I^\star_{k-1}
\end{equation}
has already been defined. If the graph $\Gamma^\star_{k-1}(\mathrsfs{A})$ is \scn, then $\Gamma^\star(\mathrsfs{A}):=\Gamma^\star_{k-1}(\mathrsfs{A})$ and the process stops with SUCCESS. If $\Gamma^\star_{k-1}(\mathrsfs{A})$ is not \scn, we proceed as follows. Given a cluster $\Delta$ of $\Gamma^\star_{k-1}(\mathrsfs{A})$, its \emph{support} is defined as the set of all vertices from $Q$ that belong to $\Delta$; the cardinality of the support is called the \emph{rank} of $\Delta$. If all clusters of $\Gamma^\star_{k-1}(\mathrsfs{A})$ have rank less than $k$, we also set $\Gamma^\star(\mathrsfs{A}):=\Gamma^\star_{k-1}(\mathrsfs{A})$ and the process stops with FAILURE. Otherwise the set $Q_k$ is defined as the collection of the supports of all clusters of rank at least $k$ in the graph $\Gamma^\star_{k-1}(\mathrsfs{A})$.  We define $\Gamma^\star_k(\mathrsfs{A})$ as the graph whose vertex set is $Q\cup Q^\star_2\cup\cdots\cup Q^\star_{k-1}\cup Q^\star_k$ and whose edge set is the union of the set \eqref{eq:e3} with the two following sets:
\[
I^\star_k:=\{(q,C)\in Q\times Q^\star_k\mid q\in C\}\cup\bigcup_{i=2}^{k-1}\{(D,C)\in Q^\star_i\times Q^\star_k\mid D\subset C\},
\]
whose edges represent inclusions between the elements of $Q\cup Q^\star_2\cup\cdots\cup Q^\star_{k-1}$ and the elements in $Q^\star_k$, and
\[
E_k:=\{(C,p)\in Q^\star_k\times Q\mid C\supseteq\excl(w),\, p\in\dupl(w)\ \text{ for some }\ w\in W_k(\mathrsfs{A})\}.
\]

Reusing our running example again, we illustrate the above construction with Fig.~\ref{fig:gstar3e5} that shows the graph $\Gamma^\star_3(\mathrsfs{E}_5)$.
\begin{figure}[hbt]
\begin{center}
\unitlength=0.95mm
\begin{picture}(105,40)(0,-5)
\node(A3)(1,0){2}
\node(B3)(25,0){1}
\node(C3)(49,0){3}
\node(D3)(76,0){4}
\node(E3)(100,0){5}
\drawedge[curvedepth=-3](A3,B3){}
\drawedge(B3,A3){}
\drawedge(C3,B3){}
\drawedge[curvedepth=3](D3,E3){}
\drawedge(E3,D3){}
\node(F3)(13,17){\{1,2\}}
\drawedge[dash={1.5}{1.5}](A3,F3){}
\drawedge[dash={1.5}{1.5}](B3,F3){}
\drawedge[curvedepth=3](F3,C3){}
\node(G3)(88,20){\{4,5\}}
\drawedge[dash={1.5}{1.5}](D3,G3){}
\drawedge[dash={1.5}{1.5}](E3,G3){}
\drawedge[curvedepth=-3](G3,C3){}
\drawedge[curvedepth=-4](G3,B3){}
\node[Nadjust=w,Nadjustdist=1.5](H3)(36,34){\{1,2,3\}}
\drawedge[dash={1.5}{1.5},curvedepth=12](A3,H3){}
\drawedge[dash={1.5}{1.5}](B3,H3){}
\drawedge[dash={1.5}{1.5}](C3,H3){}
\drawedge[dash={1.5}{1.5}](F3,H3){}
\drawedge[curvedepth=1](H3,D3){}
\drawedge[curvedepth=4](H3,E3){}
\end{picture}
\caption{The graph $\Gamma^\star_3(\mathrsfs{E}_5)$}\label{fig:gstar3e5}
\end{center}
\end{figure}

Comparing~\cite[Theorem~5]{BondarVolkov18} with Theorem~\ref{thm:necess&suffic} of the present article shows that the constructions $\Gamma^\star(\mathrsfs{A})$ and $\Gamma(\mathrsfs{A})$  are equivalent in a sense: the former graph is \scn\ if and only if so is the latter graph. Moreover, analyzing the proofs in \cite[Section 4]{BondarVolkov18} and those in Subsection~\ref{subsec:main}, one can see that the number of steps needed to build each of the two graphs is exactly the same for every DFA $\mA$. On the other hand, constructing $\Gamma(\mathrsfs{A})$ uses tactics opposite to the one utilized for building $\Gamma^\star(\mathrsfs{A})$. When passing from $\Gamma_{k-1}(\mathrsfs{A})$ to $\Gamma_{k}(\mathrsfs{A})$, we invoke condensation\footnote{The authors thank Pedro V. Silva for kindly sharing with them his insightful guess that using condensation could simplify the graph-theoretical characterization of complete reachability from \cite{BondarVolkov18}.} so that $\Gamma_{k}(\mathrsfs{A})$ always has no more vertices than $\Gamma_{k-1}(\mathrsfs{A})$. (We mention that the number of vertices can stay the same: the process described in Subsection~\ref{subsec:gammak} may have `idle' steps at which words of certain defect force no new edges.) In contrast, $\Gamma^\star_{k}(\mathrsfs{A})$ always has more vertices than $\Gamma^\star_{k-1}(\mathrsfs{A})$ since we append certain clusters of $\Gamma^\star_{k-1}(\mathrsfs{A})$ as new vertices. Therefore, our present construction is definitely more succinct. The reduction in size can be essential even for small automata as one sees comparing Fig.~\ref{fig:g2e5} and~\ref{fig:gstar2e5} or Fig.~\ref{fig:g3e5} and~\ref{fig:gstar3e5}. Also, no inclusion edges are necessary for the present approach because we store the information about inclusions between clusters separately, using the forest of clusters. (Comparing Fig.~\ref{fig:f3e5} and~\ref{fig:gstar3e5} provides an illustration: one readily sees that the forest $\mathcal{F}_3(\mathrsfs{E}_5)$ in Fig.~\ref{fig:f3e5} encodes precisely the same inclusions that are shown with dashed arrows in Fig.~\ref{fig:gstar3e5}.) We believe that separating the data describing inclusions between clusters from the data that represent relationships between disjoint clusters lightens the construction not only technically but also conceptually.

\section{Algorithmic issues}
\label{sec:algortihm}
The discussion in this section assumes the reader's acquaintance with some basics of combinatorial algorithms. These basics can be found, e.g., in the textbook \cite{Cormen}.

To decide whether or not a DFA $\mathrsfs{A}=\langle Q,\Sigma,\delta\rangle$ is completely reachable, one can use its \emph{powerset automaton} $\mathcal{P}(\mathrsfs{A}):=\langle\mathcal{P}(Q),\Sigma,\delta\rangle$. (Recall that $\mathcal{P}(Q)$ denotes the set of all non-empty subsets of $Q$ and $\delta(P,a):=\{\delta(q,a)\mid q\in P\}$ for every pair $(P,a)\in\mathcal{P}(Q)\times\Sigma$.) The definition of complete reachability readily implies that $\mA$ is completely reachable if and only if every $P\in\mathcal{P}(Q)$ is reachable from $Q$ in the underlying graph of  $\mathcal{P}(\mathrsfs{A})$. The latter property can be easily recognized by breadth-first search (BFS). BFS in a graph can be performed in linear time of the sum of its vertex and edge numbers; see~\cite[Section~22.2]{Cormen}. However, the vertex and edge numbers of the underlying graph of $\mathcal{P}(\mathrsfs{A})$ are exponential functions of the size of $\mathrsfs{A}$, whence the outlined procedure requires exponential time.

It is natural to ask whether or not complete reachability can be decided in polynomial time. Recall, for comparison, that the property of being synchronizing is polynomially decidable; see, e.g., \cite[Proposition 2.1]{KV}.

The characterization provided by Theorem~\ref{thm:necess&suffic} shows that an algorithmically efficient characterization of complete reachability would be possible if constructing the graph $\Gamma(\mathrsfs{A})$ could be done in time bounded by a polynomial of the size of $\mathrsfs{A}$. It is still open whether or not such a polynomial time construction exists. In the iterative procedure presented in Subsections~\ref{subsec:gamma1} and \ref{subsec:gammak}, two sorts of half-steps alternate: on the odd half-steps, we append edges forced by words of a given defect, while on the even half-steps, we condensate the resulting graph. The condensation of a given graph can be built in linear time of the sum of its vertex and edge numbers; see~\cite[Section~22.5]{Cormen}. Therefore, the condensation steps make no problem. In a simple graph, every edge is uniquely determined by its source and target. Hence the number of edges that can be appended when we build the graph $\Gamma_k(\mathrsfs{A})$ is at most quadratic in its vertex number, and the latter number does not exceed the number of states in $\mathrsfs{A}$. However, in order to decide which edges are to be appended, one has to analyze all transformations caused by words of defect $k$ with respect to $\mA$, and for $\mA$ with $n$ states, the number of such transformations may be as large as the product $\stirling{n}{n-k}\binom{n}{k}k!$ where the first factor is the Stirling number of the second kind. That is why implementing the odd half-steps of our procedure in polynomial time constitutes a nontrivial task.

As mentioned, Gonze and Jungers \cite{GJ:2019} have devised an algorithm that, given a DFA $\mA$ with $n$ states and $m$ input letters, constructs the graph $\Gamma_1(\mA)$ in time bounded by a polynomial in $m$ and $n$. Here we extend their ideas to a polynomial time algorithm that, given a DFA $\mA$, constructs the graph $\Gamma_k(\mA)$ for any \textbf{fixed} $k$.

Recall that $W_k(\mathrsfs{A})$ stands for the set of all words of defect~$k$ with respect to $\mathrsfs{A}$. Set
\begin{equation}
\label{eq:xd}
XD_k(\mA):=\{(\excl(w),\dupl(w))\mid w\in W_k(\mathrsfs{A})\}.
\end{equation}
Observe that the set $XD_1(\mA)$ is nothing but the edge set $E_1$ of $\Gamma_1(\mA)$ as defined in \eqref{eq:e1}. For $k>1$, the set $XD_k(\mA)$ differs from the set $E_k$ defined in \eqref{eq:ek} but the latter set can be readily recovered whenever the former one is known. Indeed, given the vertex set $Q_k$ of the graph $\Gamma_k(\mA)$, one simply selects for each $(X,D)\in XD_k(\mA)$, all pairs $(C,C')\in Q_k\times Q_k$ such that $X\subseteq\lf(C)$ and $D\cap\lf(C')\ne\varnothing$. Thus, the set $E_k$ can be found in polynomial in $|Q|$ time, provided a polynomial in $|Q|$ upper bound on the size of the set $XD_k(\mA)$. Now we deduce such a bound for each fixed $k$.

Let $\mA =\langle Q,\Sigma\rangle$ and $n:=|Q|$. If $w$ is any word of defect~$k$, then $\excl(w)$ is a $k$-element subset of $Q$. As for $\dupl(w)$, it is contained in the $(n-k)$-element set $Q{\setminus}Q\dt w$, and it is easy to see that the size of $\dupl(w)$ may vary from 1 to $\min\{k,n-k\}$.  Hence
\begin{equation}
\label{eq:bound}
|XD_k(\mA)|\le\binom{n}{k}\cdot\sum_{d=1}^{\min\{k,n-k\}}\binom{n-k}{d}.
\end{equation}
If $k\ge n-k$, the second factor of the right-hand side of \eqref{eq:bound} is $\sum_{d=1}^{n-k}\binom{n-k}{d}=2^{n-k}-1<2^k$ so that $|XD_k(\mA)|<2^k\binom{n}{k}$. If $k<n-k$, the second factor of the right-hand side of \eqref{eq:bound} is $\sum_{d=1}^k\binom{n-k}{d}$. We show that the latter sum is less than $\binom{n}{k}$ by induction on $k$. If $k=1$,  the sum reduces to $\binom{n-1}{1}=n-1<n=\binom{n}{1}$. Now assume that $n-k>k>1$ and represent $\sum_{d=1}^k\binom{n-k}{d}$ as
\[
\sum_{d=1}^k\binom{n-k}{d}=\sum_{d=1}^{k-1}\binom{(n-1)-(k-1)}{d}+\binom{n-k}{k}.
\]
Here we can apply the induction hypothesis to the first summand of the right-hand side (with $n-1$ in the role of $n$), getting $\sum_{d=1}^{k-1}\binom{(n-1)-(k-1)}{d}<\binom{n-1}{k-1}$. Hence
\[
\sum_{d=1}^k\binom{n-k}{d}<\binom{n-1}{k-1}+\binom{n-k}{k}<\binom{n-1}{k-1}+\binom{n-1}{k}=\binom{n}{k},
\]
using the obvious fact that a binomial coefficient increases when so does its top argument and Pascal's rule. Thus, we have $|XD_k(\mA)|<\left[\binom{n}{k}\right]^2$.
Since the binomial coefficient $\binom{n}{k}$ viewed as a polynomial in $n$ has degree $k$, we have established the following fact:
\begin{lemma}
\label{lem:estimate}
For any DFA $\mA$ with $n$ states and any fixed $k<n$, the cardinality of the set $XD_k(\mA)$ defined in \eqref{eq:xd} is upper bounded by a polynomial of $n$ of degree $\le 2k$.
\end{lemma}

The estimate of Lemma~\ref{lem:estimate} is one of the two facts ensuring the polynomiality of our algorithm. The other crucial fact is that one can select a prefix-closed set of words that generate all pairs $(\excl(w),\dupl(w))\in XD_k(\mA)$. Recall that a word $u$ is said to be a \emph{prefix} of another word $w$ if $w=uv$ for some word $v$. A set of words is \emph{prefix-closed} if it contains all prefixes of each its word.

We need an elementary property of transformations of a finite set. Since we are interested only in transformations of the state set $Q$ of a DFA $\mA = \langle Q,\Sigma\rangle$ induced by words from $\Sigma^*$, we state the property in this setting. For a word $v\in Q$ and a state $q\in Q$, denote by $qv^{-1}$ the preimage of $q$ under the transformation induced by $v$, that is, $$qv^{-1}:=\{p\in Q\mid p\dt v=q\}.$$

\begin{lemma}
\label{lem:folklore}
For all words $u,v\in\Sigma^*$,
\begin{align}
\label{eq:exclprod}\excl(uv)&=\{q\in Q\mid qv^{-1}\subseteq\excl(u)\},\\
\label{eq:duplprod}\dupl(uv)&=\{q\in Q\mid qv^{-1}\cap\dupl(u)\ne\varnothing\ \text{ or }\ |qv^{-1}{\setminus}\excl(u)\}|\ge 2\}.
\end{align}
\end{lemma}

We omit the proof because the equalities \eqref{eq:exclprod} and \eqref{eq:duplprod} become clear as soon as the definitions of $\excl(\ )$ and $\dupl(\ )$ are deciphered. Fig.~\ref{fig:product} illustrates the essence of the lemma.

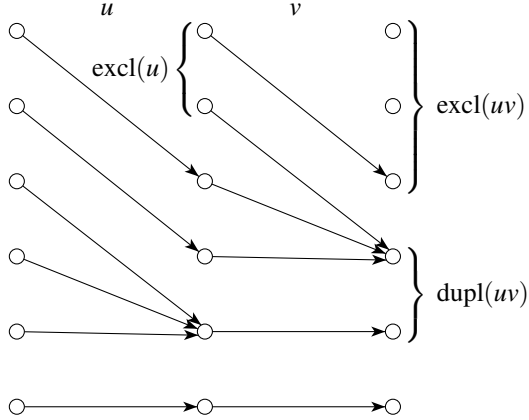
\begin{figure}[ht]
\begin{center}
\unitlength=1mm
\linethickness{0.4pt}
\gasset{Nw=2,Nh=2,AHLength=2,AHdist=3}
\begin{picture}(60,52.00)(0,5)
\node(A1)(5.00,5.00){}
\node(A2)(5.00,15.00){}
\node(A3)(5.00,25.00){}
\node(A4)(5.00,35.00){}
\node(A5)(5.00,45.00){}
\node(A6)(5.00,55.00){}
\node(B1)(30.00,5.00){}
\node(B2)(30.00,15.00){}
\node(B3)(30.00,25.00){}
\node(B4)(30.00,35.00){}
\node(B5)(30.00,45.00){}
\node(B6)(30.00,55.00){}
\node(C1)(55.00,5.00){}
\node(C2)(55.00,15.00){}
\node(C3)(55.00,25.00){}
\node(C4)(55.00,35.00){}
\node(C5)(55.00,45.00){}
\node(C6)(55.00,55.00){}
\put(17.00,57.00){\makebox(0,0)[cb]{$u$}}
\put(42.00,57.00){\makebox(0,0)[cb]{$v$}}
\put(56,19){$\left.\rule{0pt}{7mm}\right\}\ \dupl(uv)$}
\put(56,44){$\left.\rule{0pt}{13mm}\right\}\ \excl(uv)$}
\put(15,49){$\excl(u)\left\{\rule{0pt}{7mm}\right.$}
\drawedge(A1,B1){}
\drawedge[eyo=-.5](A2,B2){}
\drawedge(A3,B2){}
\drawedge[eyo=.5](A4,B2){}
\drawedge(A5,B3){}
\drawedge(A6,B4){}
\drawedge(B1,C1){}
\drawedge(B2,C2){}
\drawedge[eyo=-.5](B3,C3){}
\drawedge(B4,C3){}
\drawedge[eyo=.5](B5,C3){}
\drawedge(B6,C4){}
\end{picture}
\end{center}
\caption{An illustration for Lemma~\ref{lem:folklore}}\label{fig:product}
\end{figure}

Lemma~\ref{lem:folklore} implies that whenever $v$ is fixed, the pair $(\excl(uv),\dupl(uv))$ does not change with $u$, provided that the pair $(\excl(u),\dupl(u))$ remains unchanged. We register this observation as follows.

\begin{corollary}
\label{cor:change}
For all words $u,u',v\in\Sigma^*$, if $(\excl(u),\dupl(u))=(\excl(u'),\dupl(u'))$, then $(\excl(uv),\dupl(uv))=(\excl(u'v),\dupl(u'v))$.
\end{corollary}

We fix a linear order $\prec$ on the input alphabet $\Sigma$ of $\mA$ and extend it to the shortlex order on $\Sigma^*$ (still denoted $\prec$). Recall the definition of the  shortlex  order: for two different words $w,w'\in\Sigma^*$, one has $w\prec w'$ if and only if either $|w|<|w'|$ or $|w|=|w'|$ and $w=uav$, $w'=ubv'$ for some $a,b\in\Sigma$ such that $a\prec b$ and some $u,v,v'\in\Sigma^*$. It is well known that $\langle\Sigma^*,\prec\rangle$ is a well-ordered set and that $\prec$ is compatible with multiplication in $\Sigma^*$, that is, $w\prec w'$ implies  $uwv\prec uw'v$ for any  $u,v'\in\Sigma^*$.

For every pair $(X,D)\in XD_k(\mA)$, we denote by $w_{X,D}$ the shortlex least word such that $X=\excl(w)$ and $D=\dupl(w)$ and let $\ov{W}_k$ be the set formed by all words $w_{X,D}$. It is convenient to introduce the set $\ov{W}_0$ consisting of the empty word only.  Now define the set $\ov{W}_{\le k}$ as $\bigcup_{\ell=0}^{k}\ov{W}_{\ell}$. Our next lemma implies that the set $\ov{W}_{\le k}$ is prefix-closed.
\begin{lemma}
\label{lem:prefix}
If $u$ is a prefix of a word in $\ov{W}_k$ and has defect $\ell\le k$, then $u$ belongs to $\ov{W}_\ell$.
\end{lemma}

\begin{proof}
Suppose that $u\notin\ov{W}_\ell$ and let $u'$ be the shortlex least word such that $\excl(u')=\excl(u)$ and $\dupl(u')=\dupl(u)$. We have $u'\in\ov{W}_\ell$ by the definition of $\ov{W}_\ell$ whence $u'\ne u$ and $u'\prec u$. Now let $w\in\ov{W}_k$ be such that $w=uv$ for some $v\in\Sigma^*$. Then $u'v\prec uv=w$, and by Corollary~\ref{cor:change} we have $(\excl(w),\dupl(w))=(\excl(u'v),\dupl(u'v))$. This is a contradiction because by the definition of $\ov{W}_k$, the word $w$ is the shortlex least among words $w'$ of defect $k$ with $(\excl(w'),\dupl(w'))=(\excl(w),\dupl(w))$.
\end{proof}

Observe that the map $(X,D)\mapsto w_{X,D}$ is a bijection. Hence $|\ov{W}_{k}|=|XD_k(\mA)|$, and Lemma~\ref{lem:estimate} implies that the cardinality of $\ov{W}_k$ is upper bounded by a polynomial of $n$ of degree $\le 2k$.
The same conclusion holds for $\ov{W}_{\le k}=\bigcup_{\ell=0}^{k}\ov{W}_{\ell}$. Thus, $\ov{W}_{\le k}$ is a prefix-closed set of words over $\Sigma$ of polynomial in $n$ size.

Now we invoke a general scheme that uses BFS for building prefix-closed sets of words with certain properties. The scheme is fairly standard but we failed to find a reference where it would be stated in a form that fully suits the usage here. Therefore, we present it in some detail for completeness as well as the reader's convenience. It operates with list of words. By a \emph{list} we mean a set whose elements are listed in a fixed linear order; we use brackets [ ] as delimiters for lists to distinguish them from sets. For disjoint lists $L=[\alpha,\beta,\gamma,\dots]$ and $L'=[\alpha',\beta',\gamma',]$, we denote the list $[\alpha,\beta,\gamma,\dots,\alpha',\beta',\gamma',\dots]$ by $L\sqcup L'$. If $L$ is a list and $\omega\notin L$, we can form the list $L\sqcup[\omega]$, in which case we say that we \emph{append} $\omega$ to $L$.

\begin{proposition}
\label{prop:scheme}
Let $\Sigma$ be a alphabet of size $m$ with a fixed linear order and let $\mathfrak{P}$ be a property of words over $\Sigma$ such that:
\begin{enumerate}
  \item[\emph{(i)}]  given a word $w\in\Sigma^*$, one can decide whether or not $w$ satisfies $\mathfrak{P}$ in time $\le T$;
  \item[\emph{(ii)}] the set $W_{\mathfrak{P}}$ of all words satisfying $\mathfrak{P}$ is prefix-closed;
  \item[\emph{(iii)}] there is a positive integer $N$ such that $|W_{\mathfrak{P}}|\le N$.
\end{enumerate}
There exists an algorithm that returns the list of all words satisfying $\mathfrak{P}$ in ascending shortlex order and takes time $\le mNT$.
\end{proposition}

\begin{proof}
We may assume that the empty word $\varepsilon$ satisfies $\mathfrak{P}$. (If not, the set $W_{\mathfrak{P}}$ is empty since $\varepsilon$ is a prefix of every word and $W_{\mathfrak{P}}$ is prefix-closed.) We use BFS as follows.

Two lists $I$ (input list) and $W$ (output list) are initialized with the list $[\varepsilon]$. The main loop executes while the input list $I$ is non-empty. In each round, we initialize yet another list $S$ (storage list) with the empty list. Then we process $I$ in the ascending shortlex order. For each word $w\in I$, we browse through the letters in $\Sigma$ in the order fixed on $\Sigma$. For each $a\in\Sigma$, we check whether the word $wa$ satisfies $\mathfrak{P}$. If not, we discard $wa$; if yes, we append $wa$ to the list $S$. As the end of the round, we replace $I$ with $S$ and $W$ with $W\sqcup S$.

The pseudocode of the described procedure \textsc{Listing}($\mathfrak{P}$) is shown in Algorithm~\ref{alg:ListP}.

\begin{algorithm}
\setlength{\commentspace}{5cm}
\begin{algorithmic}[1]
\FUNC{Listing($\mathfrak{P}$)}
\IF {$\varepsilon$ does not satisfy $\mathfrak{P}$}
\STATE $W\gets[\ ]$
\ENDIF
\STATE\algcomment{1}{Initializing the output list}$W\gets[\varepsilon]$
\STATE\algcomment{1}{Initializing the input list}$I\gets[\varepsilon]$
\WHILE{$I\ne[\ ]$}
\STATE\algcomment{2}{Initializing the storage list}$S\gets[\ ]$
\FOR{$w\in I$\algcomment{5.5}{$I$ is processed in the ascending shortlex order}}
\FOR{$a\in\Sigma$\algcomment{6.5}{$\Sigma$ is processed in the ascending shortlex order}}
\IF  {$wa$ satisfies $\mathfrak{P}$}
\STATE append $wa$ to $S$
\ENDIF
\ENDFOR
\ENDFOR
\STATE\algcomment{1.6}{Updating the output list}$W\gets W\sqcup S$
\STATE\algcomment{1.6}{Updating the input list}$I\gets S$
\ENDWHILE
\RETURN $W$
\end{algorithmic}
\caption{Computing the list of words that satisfy a given property $\mathfrak{P}$}\label{alg:ListP}
\end{algorithm}

Now we show that the output $W$ of \textsc{Listing}($\mathfrak{P}$) is the desired shortlex sorted list of all words satisfying $\mathfrak{P}$. Indeed, words appended to $W$ during round $\ell$ of the main loop (lines 5--12 of the pseudocode) are of the form $wa$ where $w\in I$ is a word appended in round $\ell-1$ and $a$ is a letter. It follows by induction on $\ell$ that all these words have length $\ell$. Therefore, words appended to $W$ during each round are longer than all words that occurred in $W$ before that round. If in a certain round, we process $w,w'\in I$ with $w\prec w'$, then $w$ is processed first whence for all $a,b\in\Sigma$, any word of the form $wa$ is appended to $W$ earlier than any word of the form $w'b$. If $a,b\in\Sigma$ are such that $a\prec b$, then for every $w\in I$, the letter $a$ is processed first and the word $wa$ is appended to $W$ earlier than the word $wb$. Altogether, we see that the list $W$ is shortlex sorted by its construction.

Also by construction, only words from the set $W_{\mathfrak{P}}$ get appended to $W$. Conversely, we prove that every $w\in W_{\mathfrak{P}}$ appears in the list $W$ by induction on $|w|$. Indeed, $\varepsilon\in W$, and
since $W_{\mathfrak{P}}$ is prefix-closed by (ii), each $w\in W_{\mathfrak{P}}$ with $|w|>0$ can be written as $w'a$ where $w'\in W_{\mathfrak{P}}$ and $a\in\Sigma$. By the inductive assumption, $w'\in W$ whence $w=w'a$ is appended to $W$ during round $|w|$ of the main loop.

It remains to verify that the procedure \textsc{Listing}($\mathfrak{P}$) stops in time $\le mNT$. Indeed, the procedure checks if $wa$ satisfies $\mathfrak{P}$ (line 9 of the pseudocode) for all pairs $(w,a)\in W_{\mathfrak{P}}\times\Sigma$. The number of such pairs is $\le mN$ by (iii) and each check requires time $\le T$ by (i).
\end{proof}

Now, for each fixed $k$, consider the property ``$w\in\ov{W}_{\le k}$''. The set of words with this property is $\ov{W}_{\le k}$, and it satisfies condition (ii) of Proposition~\ref{prop:scheme} by Lemma~\ref{lem:prefix} and condition (iii) of Proposition~\ref{prop:scheme} with $N=O(n^{2k})$ by Lemma~\ref{lem:estimate}. Thus, Proposition~\ref{prop:scheme} provides an algorithm that lists all words of $\ov{W}_{\le k}$ in shortlex order in time $O(mn^{2k}T)$ where $T$ is an upper bound on time for the check in line 9 of the pseudocode. In each round of the main loop, this check amounts to
\begin{itemize}
  \item[(A)] computing for every $w\in I$ and $a\in\Sigma$, the sets $X:=\excl(wa)$ and $D:=\dupl(wa)$, and
  \item[(B)] comparing the pair $(X,D)$ with all pairs $(\excl(u),\dupl(u))$ such that the word $u$ has already been stored in either the output list $W$ or the storage list $S$.
\end{itemize}
Indeed, as shown in the proof of Proposition~\ref{prop:scheme}, words are appended to $W$ and $S$ in shortlex order whence $u\prec wa$ for all $u$ that already appear in either $W$ or $L$. Therefore, if $(X,D)=(\excl(u),\dupl(u))$ for some such $u$, then $wa$ is not the shortlex least with these excluded and duplicate sets, and thus, $wa\notin\ov{W}_{\le k}$. If $(X,D)\ne(\excl(u),\dupl(u))$ for every such $u$, then $wa=w_{X,D}\in\ov{W}_{\le k}$ since $wa\prec v$ for each $v$ that gets appended to $W$ and $S$ later.

It remains to efficiently implement the computation in (A) and the comparison in (B). For this, we modify the general procedure in Algorithm~\ref{alg:ListP}: the lists $W$, $I$, and $S$ will now consist of triples of the form $(u,\excl(u),\dupl(u))$, that is, with each word $u$, we will store its pair $(\excl(u),\dupl(u))$. The lists themselves will be stored as an appropriate data structures, say, self-balancing binary search trees (cf.\ \cite[Chapter 13]{Cormen}). Then, given a pair $(X,D)$, one can perform the comparison in (B) in time logarithmic in the size of the lists $W$ and $S$. At any step of the algorithm, the sizes of $W$ and $S$ do not exceed the final size of $W$, that is, $|\ov{W}_{\le k}|=O(n^{2k})$. Therefore, the comparison in (B) can be done in $O(\log n)$ time. As for (A), Lemma~\ref{lem:folklore} shows how $(\excl(wa),\dupl(wa))$ can be computed from the pair $(\excl(w),\dupl(w))$ and the collection $\{qa^{-1}\mid q\in Q\}$. The latter can of course be pre-computed for all $a\in\Sigma$ and stored as appropriate data structure. The rules \eqref{eq:exclprod} and \eqref{eq:duplprod} involve the subsets $\excl(w)$, $\dupl(w)$ of $Q$ that have size at most $k$. Since $k$ is fixed and $|Q|=n$, these rules can be implemented in $O(\log n)$ time. We conclude that with the modifications outlined, each check in line~9 of the pseudocode can be done in $O(\log n)$ time.

Once the modified version of the output list $W$ has been computed, we can extract the triples $(u,\excl(u),\dupl(u))$ in which the word $u$ has defect $k$. The pairs $(\excl(u),\dupl(u))$ from these triples form the set $XD_k(\mA)$. We have already outlined how the edges of $E_k$ can be reconstructed from $XD_k(\mA)$; now we estimate time needed for the reconstruction. Recall that by \eqref{eq:ek}, the edges in $E_k$ are the pairs $(C,C')\in Q_k\times Q_k$ such that $X\subseteq \lf(C)$ and $D\cap\lf(C')\ne\varnothing$ for some $(X,D)\in XD_k(\mA)$; here $Q_k$ is the vertex set of the graph $\Gamma_k(\mA)$. We use the same tactics as in dealing with the task (A) above: we pre-compute the partition $\{\lf(C)\mid C\in Q_k\}$ of the set $Q$ and then check the conditions $X\subseteq \lf(C)$ and $D\cap\lf(C')\ne\varnothing$ for each $(X,D)\in XD_k(\mA)$. We have $|X|=k$ and $|D|\le k$ for all $(X,D)\in XD_k(\mA)$, and since $k$ is fixed, each such check can be implemented in $O(\log n)$ time. Using the estimate $|XD_k(\mA)|=O(n^{2k})$, we conclude that the whole set $E_k$ can be computed in time $O(n^{2k}\log n)$.

Summarizing, we have the following result.

\begin{theorem}
\label{thm:algorithm}
Let $\mA$ be a DFA with $n$ states and $m$ input letters. For each $k<n$, there exists an algorithm that builds the graph $\Gamma_k(\mA)$ in time $O(mn^{2k}\log n)$.
\end{theorem}

\section{The number of steps in the construction of $\Gamma(\mathrsfs{A})$}
\label{sec:applications}

The following question is both natural and important: is there an absolute constant $K$ such that for each DFA $\mA$, the construction of the graph $\Gamma(\mathrsfs{A})$ terminates after at most $K$ steps? Thanks to Theorem~\ref{thm:algorithm}, an affirmative answer to this question would ensure a polynomial time algorithm for recognizing complete reachability. However, as mentioned at the end of Subsection~\ref{subsec:gammak}, no such absolute constant exists. Here, for each $n,k$ with $2\le k<n$, we exhibit a completely reachable DFA $\mathrsfs{E}_{n,k}$ with $n$ states such that the construction of each of the graphs $\Gamma(\mathrsfs{E}_{n,k})$ requires exactly $k$ steps. (A similar result was announced in~\cite[Section~5]{BondarVolkov18}, but the series of examples presented there works as desired only for $k=n-1$.)

We assume that the state set of $\mathrsfs{E}_{n,k}$ is the set $Q=\{1,2\dots,n\}$. The automaton $\mathrsfs{E}_{n,k}$ has $n+k-1$ input letters. They come in two groups: $a_1,\dots,a_n$ and $b_\ell,\dots,b_{n-1}$ where $\ell$ stands for $n-k+1$; observe that $2\le\ell\le n-1$. The action of the letters $b_\ell,\dots,b_{n-1}$ is defined as follows: for $q\in Q$ and $i\in\{\ell,\dots,n-1\}$,
\begin{equation}
\label{eq:actionb}
q\dt b_{i}:=\begin{cases}q&\text{if $1<q<\ell$ or $q>i$},\\ i+1&\text{if $q=1$ or $\ell\le q\le i$}. \end{cases}
\end{equation}
The action of the letters $a_1,\dots,a_n$ is defined as follows: for $q\in Q$ and $j\in\{1,2,\dots,n\}$,
\begin{align}
\label{eq:actiona}
\text{if }\ j<\ell,\ \text{then }\ q\dt a_{j}&:=\begin{cases}q&\text{if $q\ne j$},\\ q+1&\text{if $q=j$};\end{cases}\notag\\
\text{if }\ j=\ell,\ \text{then }\ q\dt a_{j}&:=\begin{cases}q&\text{if $q\ne\ell$},\\ 1&\text{if $q=\ell$};\end{cases}\\
\text{if }\ j>\ell,\ \text{then }\ q\dt a_{j}&:=\begin{cases}q&\text{if $q<\ell$ or $q>j$},\\ 1&\text{if $q=\ell$},\\ q-1&\text{if $\ell<q\le j$}.\end{cases}\notag
\end{align}
Fig.~\ref{fig:actiona&b} illustrates definitions \eqref{eq:actionb} and \eqref{eq:actiona}.

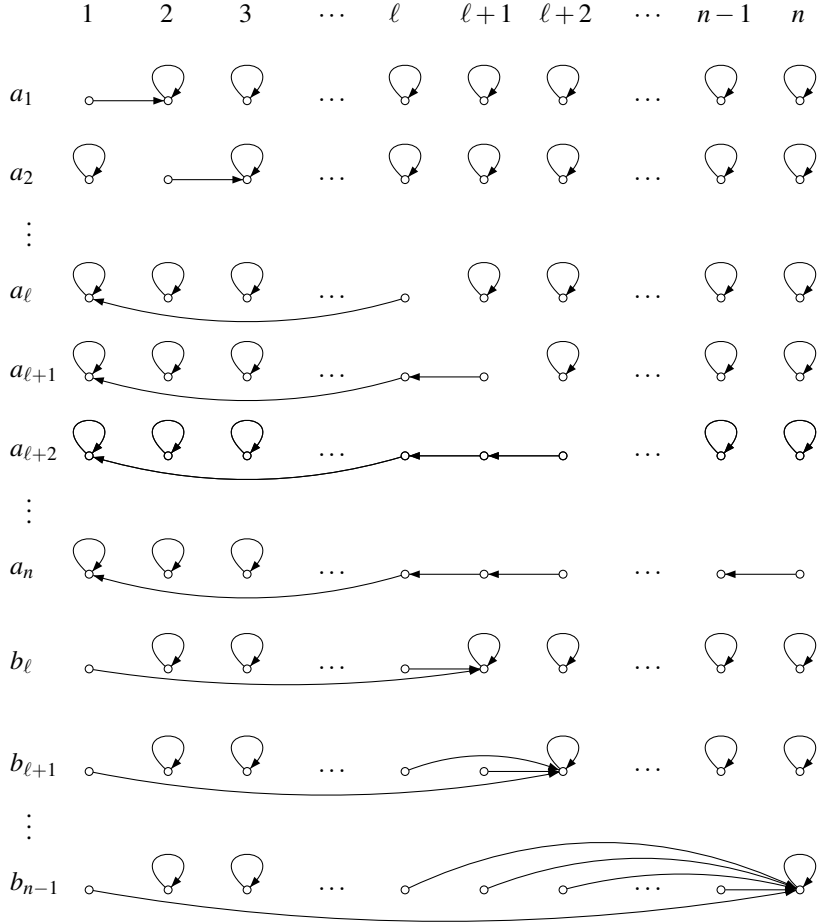
\begin{figure}[bt]
\begin{center}
\unitlength=1.05mm
\begin{picture}(120,120)(0,-25)
\put(9,90){1}
\put(19,90){2}
\put(29,90){3}
\put(39,90){$\cdots$}
\put(39,80){$\dots$}
\put(39,70){$\dots$}
\put(39,55){$\dots$}
\put(39,45){$\dots$}
\put(39,35){$\dots$}
\put(39,20){$\dots$}
\put(39,8){$\dots$}
\put(39,-5){$\dots$}
\put(39,-20){$\dots$}
\put(48,90){$\ell$}
\put(57,90){$\ell+1$}
\put(67,90){$\ell+2$}
\put(79,90){$\cdots$}
\put(79,80){$\dots$}
\put(79,70){$\dots$}
\put(79,55){$\dots$}
\put(79,45){$\dots$}
\put(79,35){$\dots$}
\put(79,20){$\dots$}
\put(79,8){$\dots$}
\put(79,-5){$\dots$}
\put(79,-20){$\dots$}
\put(87,90){$n-1$}
\put(99,90){$n$}
\put(0,80){$a_1$}
\put(0,70){$a_2$}
\put(2,61.5){$\vdots$}
\put(0,55){$a_\ell$}
\put(0,45){$a_{\ell+1}$}
\put(0,35){$a_{\ell+2}$}
\put(2,26.5){$\vdots$}
\put(0,20){$a_n$}
\put(0,8){$b_\ell$}
\put(0,-5){$b_{\ell+1}$}
\put(2,-13.5){$\vdots$}
\put(0,-20){$b_{n-1}$}
\gasset{Nw=1,Nh=1,loopdiam=4}
\node(A11)(10,80){}
\node(A12)(20,80){}
\node(A13)(30,80){}
\node(A15)(50,80){}
\node(A16)(60,80){}
\node(A17)(70,80){}
\node(A19)(90,80){}
\node(A20)(100,80){}
\drawedge(A11,A12){}
\drawloop(A12){}
\drawloop(A13){}
\drawloop(A15){}
\drawloop(A16){}
\drawloop(A17){}
\drawloop(A19){}
\drawloop(A20){}
\node(A21)(10,70){}
\node(A22)(20,70){}
\node(A23)(30,70){}
\node(A25)(50,70){}
\node(A26)(60,70){}
\node(A27)(70,70){}
\node(A29)(90,70){}
\node(A30)(100,70){}
\drawedge(A22,A23){}
\drawloop(A21){}
\drawloop(A23){}
\drawloop(A25){}
\drawloop(A26){}
\drawloop(A27){}
\drawloop(A29){}
\drawloop(A30){}
\node(A31)(10,55){}
\node(A32)(20,55){}
\node(A33)(30,55){}
\node(A35)(50,55){}
\node(A36)(60,55){}
\node(A37)(70,55){}
\node(A39)(90,55){}
\node(A40)(100,55){}
\drawloop(A31){}
\drawloop(A32){}
\drawloop(A33){}
\drawedge[curvedepth=3](A35,A31){}
\drawloop(A36){}
\drawloop(A37){}
\drawloop(A39){}
\drawloop(A40){}
\node(A41)(10,45){}
\node(A42)(20,45){}
\node(A43)(30,45){}
\node(A45)(50,45){}
\node(A46)(60,45){}
\node(A47)(70,45){}
\node(A49)(90,45){}
\node(A50)(100,45){}
\drawloop(A41){}
\drawloop(A42){}
\drawloop(A43){}
\drawedge[curvedepth=3](A45,A41){}
\drawedge(A46,A45){}
\drawloop(A47){}
\drawloop(A49){}
\drawloop(A50){}
\node(A51)(10,35){}
\node(A52)(20,35){}
\node(A53)(30,35){}
\node(A55)(50,35){}
\node(A56)(60,35){}
\node(A57)(70,35){}
\node(A59)(90,35){}
\node(A60)(100,35){}
\drawloop(A51){}
\drawloop(A52){}
\drawloop(A53){}
\drawedge[curvedepth=3](A55,A51){}
\drawedge(A56,A55){}
\drawedge(A57,A56){}
\drawloop(A59){}
\drawloop(A60){}
\node(A51)(10,35){}
\node(A52)(20,35){}
\node(A53)(30,35){}
\node(A55)(50,35){}
\node(A56)(60,35){}
\node(A57)(70,35){}
\node(A59)(90,35){}
\node(A60)(100,35){}
\drawloop(A51){}
\drawloop(A52){}
\drawloop(A53){}
\drawedge[curvedepth=3](A55,A51){}
\drawedge(A56,A55){}
\drawedge(A57,A56){}
\drawloop(A59){}
\drawloop(A60){}
\node(A61)(10,20){}
\node(A62)(20,20){}
\node(A63)(30,20){}
\node(A65)(50,20){}
\node(A66)(60,20){}
\node(A67)(70,20){}
\node(A69)(90,20){}
\node(A70)(100,20){}
\drawloop(A61){}
\drawloop(A62){}
\drawloop(A63){}
\drawedge[curvedepth=3](A65,A61){}
\drawedge(A66,A65){}
\drawedge(A67,A66){}
\drawedge(A70,A69){}
\node(B11)(10,8){}
\node(B12)(20,8){}
\node(B13)(30,8){}
\node(B15)(50,8){}
\node(B16)(60,8){}
\node(B17)(70,8){}
\node(B19)(90,8){}
\node(B20)(100,8){}
\drawedge[curvedepth=-2](B11,B16){}
\drawedge(B15,B16){}
\drawloop(B12){}
\drawloop(B13){}
\drawloop(B16){}
\drawloop(B17){}
\drawloop(B19){}
\drawloop(B20){}
\node(B21)(10,-5){}
\node(B22)(20,-5){}
\node(B23)(30,-5){}
\node(B25)(50,-5){}
\node(B26)(60,-5){}
\node(B27)(70,-5){}
\node(B29)(90,-5){}
\node(B30)(100,-5){}
\drawedge[curvedepth=-3](B21,B27){}
\drawedge[curvedepth=2](B25,B27){}
\drawedge(B26,B27){}
\drawloop(B22){}
\drawloop(B23){}
\drawloop(B27){}
\drawloop(B29){}
\drawloop(B30){}
\node(B31)(10,-20){}
\node(B32)(20,-20){}
\node(B33)(30,-20){}
\node(B35)(50,-20){}
\node(B36)(60,-20){}
\node(B37)(70,-20){}
\node(B39)(90,-20){}
\node(B40)(100,-20){}
\drawedge[curvedepth=-4](B31,B40){}
\drawedge[curvedepth=6](B35,B40){}
\drawedge[curvedepth=4](B36,B40){}
\drawedge[curvedepth=2](B37,B40){}
\drawedge(B39,B40){}
\drawloop(B32){}
\drawloop(B33){}
\drawloop(B40){}
\end{picture}
\caption{Actions of the input letters in the automaton $\mathrsfs{E}_{n,k}$}\label{fig:actiona&b}
\end{center}
\end{figure}

\begin{proposition}
\label{prop:ksteps}
For each $n,k$ with $2\le k<n$, one has $\Gamma(\mathrsfs{E}_{n,k})=\Gamma_k(\mathrsfs{E}_{n,k})$. The graph $\Gamma(\mathrsfs{E}_{n,k})$ is \scn.
\end{proposition}

\begin{proof}
By \eqref{eq:actiona}, each letter $a_j$, $j=1,\dots,n$, has defect 1 and $\excl(a_j)=j$. Further, by \eqref{eq:actionb}, $\excl(b_i)=\{1,\ell,\dots,i\}$ for each $i=\ell,\dots,n-1$ so that all these letters have defect at least~2. Therefore, words in $W_1(\mathrsfs{E}_{n,k})$, that is, words of defect 1 with respect to $\mathrsfs{E}_{n,k}$, are products of the letters $a_1,\dots,a_n$.

By \eqref{eq:actiona}, $\dupl(a_j)=j+1$ for all $j<\ell$ and $\dupl(a_j)=1$ for all $j\le\ell$. This implies that the graph $\Gamma_1(\mathrsfs{E}_{n,k})$ contains the cycle $1\to 2\to\cdots\to\ell\to1$ whose edges are forced by $a_1,\dots,a_\ell$ and the edges $j\to 1$ for all $j>\ell$ that are forced by $a_{\ell+1},\dots,a_n$. (Recall that a \emph{cycle} in a graph is a path that ends at its starting vertex.)

In the rest of the proof, we repeatedly use the following easy fact.

\begin{lemma}
\label{lem:duplicates}
Let $\mA=\langle Q,\Sigma\rangle$ be a DFA. If $P\subseteq Q$ is such that $\dupl(a)\subseteq P$ and $P\dt a\subseteq P$ for all letters $a\in\Sigma$, then $\dupl(w)\subseteq P$ for all words $w\in\Sigma^*$.
\end{lemma}

\begin{proof}
If $\dupl(w)=\varnothing$, there is nothing to prove. Thus, we assume that $\dupl(w)\ne\varnothing$ and induct on $|w|$. If $|w|=1$, then $w$ is a letter so that the claim is a part of the premise.

Let $|w|>1$. Then $w=ua$ for some word $u$ of length $|w|-1$ and some letter $a\in\Sigma$. By \eqref{eq:duplprod}, every $q\in\dupl(w)$ satisfies either $qa^{-1}\cap\dupl(u)\ne\varnothing$ or $|qa^{-1}{\setminus}\excl(u)\}|\ge2$. In the first case, $q=p\dt a$ for some $p\in\dupl(u)$. By the induction assumption, we have $p\in P$ whence $q\in P\dt a\subseteq P$. In the second case, there are some distinct $r_1,r_2\in qa^{-1}{\setminus}\excl(u)$ so that $q=r_1\dt a=r_2\dt a$ whence $p\in\dupl(a)\subseteq P$. We see that $q\in P$ in either case.
\end{proof}

In $\mathrsfs{E}_{n,k}$, the set $P_\ell:=\{1,\dots,\ell\}$ is closed under the action of the letters $a_1,\dots,a_n$ and contains the duplicate states of all these letters. Thus, Lemma~\ref{lem:duplicates} applies to the DFA $\langle Q,\{a_1,\dots,a_n\}\rangle$, yielding $\dupl(w)\subseteq P_\ell$ for all words $w\in\{a_1,\dots,a_n\}^*$. In particular, $\dupl(w)\in P_\ell$ for every word $w\in W_1(\mathrsfs{E}_{n,k})$. In terms of the graph $\Gamma_1(\mathrsfs{E}_{n,k})$, this means that the target of every edge of $\Gamma_1(\mathrsfs{E}_{n,k})$ lies in $P_\ell$. Combining this with the observation made prior to Lemma~\ref{lem:duplicates}, we conclude that $P_\ell$ is the only non-singleton cluster of $\Gamma_1(\mathrsfs{E}_{n,k})$ while each other cluster $\{j\}$, $j=\ell+1,\dots,n$, serves as the source of a single edge $\{j\}\to P_\ell$ in the condensation $\Gamma_1^{\mathsf{con}}$ of $\Gamma_1(\mathrsfs{E}_{n,k})$. Thus, $\Gamma_1^{\mathsf{con}}$ has $n-\ell+1=k$ vertices and $k-1$ edges as shown in Fig.~\ref{fig:condensation}.
\begin{figure}[hbt]
\begin{center}
\unitlength=.95mm
\begin{picture}(80,32)(0,-5)
\gasset{Nadjust=wh}
\node(A1)(0,20){$\{\ell+1\}$} \node(B1)(25,20){$\{\ell+2\}$} \put(50,20){$\cdots$} \node(C1)(75,20){$\{n\}$} \node(PL)(40,0){$\{1,2,\dots,\ell\}$}
\drawedge(A1,PL){}
\drawedge(B1,PL){}
\drawedge(C1,PL){}
\end{picture}
\caption{The condensation $\Gamma_1^{\mathsf{con}}$ of the graph $\Gamma_1(\mathrsfs{E}_{n,k})$}\label{fig:condensation}
\end{center}
\end{figure}

In order to construct the graph $\Gamma_2(\mathrsfs{E}_{n,k})$, one augments the edge set of $\Gamma_1^{\mathsf{con}}$ by the set $E_2$ defined via~\eqref{eq:e2}. Each edge in $E_2$ is of the form $(C,C')$ where $C\ne C'$ are clusters of $\Gamma_1(\mathrsfs{E}_{n,k})$ such that $C\supseteq\excl(w)$ and $C'\cap\dupl(w)\ne\varnothing$ for a word $w\in W_2(\mathrsfs{E}_{n,k})$. Singleton clusters contain no 2-element subsets so that $P_\ell$ is the only cluster that can serve as the source of an edge. Thus, we are looking for words $w\in W_2(\mathrsfs{E}_{n,k})$ such that $\excl(w)\subseteq P_\ell$ but $\dupl(w)\nsubseteq P_\ell$. One such word is $b_\ell$ since $\excl(b_\ell)=\{1,\ell\}\subseteq P_\ell$ while $\dupl(b_\ell)=\{\ell+1\}\nsubseteq P_\ell$; it forces the edge $P_\ell\to\{\ell+1\}$. We are going to verify that this is the only edge in $E_2$.

We have already observed that $\excl(b_i)=\{1,\ell,\dots,i\}$ for each $i=\ell,\dots,n-1$ whence the defect of each letter $b_i$ with $i>\ell$ is at least~3. Therefore words in $W_2(\mathrsfs{E}_{n,k})$ are products of the letters $a_1,\dots,a_n,b_\ell$. In $\mathrsfs{E}_{n,k}$, the set $P_{\ell+1}:=P_\ell\cup\{\ell+1\}$ is closed under the action of the letters $a_1,\dots,a_n,b_\ell$ and contains the duplicate states of all these letters.  Hence we can apply Lemma~\ref{lem:duplicates} to the DFA $\langle Q,\{a_1,\dots,a_n,b_\ell\}\rangle$, getting that $\dupl(w)\subseteq P_{\ell+1}$ for all words $w\in\{a_1,\dots,a_n,b_\ell\}^*$. In particular, $\dupl(w)\in P_{\ell+1}$ for every word $w\in W_2(\mathrsfs{E}_{n,k})$. This implies that the edge $P_\ell\to\{\ell+1\}$ is indeed the only possible edge in $E_2$.

Since $\Gamma_1^{\mathsf{con}}$ has the edge $\{\ell+1\}\to P_\ell$, we see that $\{P_\ell,\{\ell+1\}\}$ is a cluster in $\Gamma_2(\mathrsfs{E}_{n,k})$ whose leafage is $P_{\ell+1}$ while all other clusters in $\Gamma_2(\mathrsfs{E}_{n,k})$ are singletons. Therefore the condensation $\Gamma_2^{\mathsf{con}}$ of $\Gamma_2(\mathrsfs{E}_{n,k})$ has $k-1$ vertices and $k-2$ edges.

The same arguments, with obvious adjustments, work for constructing $\Gamma_3(\mathrsfs{E}_{n,k})$, $\Gamma_4(\mathrsfs{E}_{n,k})$, etc. When building the next graph from the condensation of the graph constructed at the previous step, Lemma~\ref{lem:duplicates} ensures that exactly one edge is added, and thus, a single 2-element cluster is created. The last condensation $\Gamma_{k-1}^{\mathsf{con}}$ has two vertices --- the clusters $D$ and $D'$ with $\lf(D)=\{1,2,\dots,n-1\}$ and $\lf(D')=\{n\}$ --- and one edge $D'\to D$. The letter $b_{n-1}$ forces the edge $D\to D'$ which makes the graph $\Gamma_k(\mathrsfs{E}_{n,k})$ \scn. Therefore, $\Gamma(\mathrsfs{E}_{n,k})=\Gamma_k(\mathrsfs{E}_{n,k})$. \hfill$\Box$\end{proof}

\medskip

If one takes $k=n-1$, Proposition~\ref{prop:ksteps} demonstrates that constructing the graph $\Gamma(\mathrsfs{A})$ for an automaton $\mA$ with $n$ states may take $n-1$ steps when the process terminates with SUCCESS. Now we present a slight modification showing that the same may happen when the process ends with FAILURE.

Let $\mathrsfs{E}'_{n,n-1}$ stand for the DFA obtained from $\mathrsfs{E}_{n,n-1}$ by omitting the letter $b_{n-1}$. The letter is of defect $n-1$ and so is every word in which $b_{n-1}$ occurs. As only words of defect $\le s$ are involved in the construction of the graph $\Gamma_s(\mathrsfs{A})$, we conclude that for all $s=1,\dots,n-2$, the graphs $\Gamma_s(\mathrsfs{E}'_{n,n-1})$ and $\Gamma_s(\mathrsfs{E}_{n,n-1})$ coincide and have the form established in the proof of Proposition~\ref{prop:ksteps}. In particular, the condensation of the graph $\Gamma_{n-2}(\mathrsfs{E}'_{n,n-1})=\Gamma_{n-2}(\mathrsfs{E}_{n,n-1})$ has two vertices --- the clusters $D$ and $D'$ with $\lf(D)=\{1,2,\dots,n-1\}$ and $\lf(D')=\{n\}$ --- and one edge $D'\to D$. It is easy to see that the state $n$ is not the duplicate state of any word over the input alphabet of $\mathrsfs{E}'_{n,n-1}$. Hence the edge $D\to D'$ does not occur in the graph $\Gamma_{n-1}(\mathrsfs{E}'_{n,n-1})$, which thus remains not \scn. Since $n-1$ is the maximum possible defect of the word, constructing the graph $\Gamma(\mathrsfs{E}'_{n,n-1})$ stops here with FAILURE.

\smallskip

The size of the input alphabets of the DFAs $\mathrsfs{E}_{n,k}$ grows with $n$ and $k$. The question of whether or not a similar series can be found amongst DFAs with restricted alphabets is more complicated. We address it in full in a separate paper, while here, we only refute a related conjecture from~\cite{BondarVolkov16}.

In~\cite{BondarVolkov16}, it was conjectured that for any DFA $\mA$ with two input letters, the strong connectivity of the graph $\Gamma_1(\mA)$ is not only sufficient but also necessary for complete reachability of $\mA$.
In order to demonstrate that this is not the case, consider the 12-state DFA $\mE_{12}=\langle \{0,1,\dots,11\},\{a,b\}\rangle$ with the action of the letters defined in Table~\ref{tb:e12}:
\begin{table}[ht]
\caption{The transition table of the automaton $\mathrsfs{E}_{12}$}\label{tb:e12}
\begin{center}
\begin{tabular}{c@{\ \  }|@{\ \  }c@{\ \  }c@{\ \  }c@{\ \  }c@{\ \  }c@{\ \  }c@{\ \  }c@{\ \  }c@{\ \  }c@{\ \  }c@{\ \  }c@{\ \  }c}
                   $q$   & 0  & 1 & 2 & 3 & 4  & 5  & 6  & 7 & 8 & 9  & 10 & 11 \mathstrut\\
\hline
                  $q\dt a$ & 10 & 1 & 2 & 8 & 4  & 3  & 10 & 9 & 5 & 7  &  6 & 11 \mathstrut\\
                  $q\dt b$ & 1  & 2 & 3 & 4 & 5  & 6  & 7  & 8 & 9 & 10 & 11 & 0
\end{tabular}
\end{center}
\end{table}

The DFA  $\mathrsfs{E}_{12}$ is shown in Fig.~\ref{fig:e12}. We see that the letter $a$ has defect 1, and 0 and 10 are its excluded and duplicate states, respectively. The letter $b$ just adds 1 modulo 12.
\begin{figure}[bht]
\begin{center}
\unitlength=.76mm
\begin{picture}(100,90)(0,21)
\put(-14,73){$\excl(a)$}
\node(A0)(6.7,75){0}
\node(A1)(25,93.30){1}
\node(A2)(50,97.30){2}
\node(A3)(75,93.30){3}
\node(A4)(93.30,75){4}
\node(A5)(100,50){5}
\node(A6)(93.30,25){6}
\node(A7)(75,40){7}
\node(A8)(50,55){8}
\node(A9)(25,40){9}
\put(-14,23){$\dupl(a)$}
\node(A10)(6.7,25){10}
\node(A11)(0,50){11}
\drawedge(A0,A10){}
\drawloop[loopangle=120](A1){}
\drawloop[loopangle=90](A2){}
\drawedge(A3,A8){}
\drawloop[loopangle=30](A4){}
\drawedge(A5,A3){}
\drawedge[curvedepth=3](A6,A10){}
\drawedge[curvedepth=3](A7,A9){}
\drawedge(A8,A5){}
\drawedge[curvedepth=3](A9,A7){}
\drawedge[curvedepth=3](A10,A6){}
\drawloop[loopangle=180](A11){}
\drawedge[dash={1.5}{.5},curvedepth=1](A11,A0){}
\drawedge[dash={1.5}{.5},curvedepth=1](A10,A11){}
\drawedge[dash={1.5}{.5},curvedepth=-1](A9,A10){}
\drawedge[dash={1.5}{.5},curvedepth=-1](A8,A9){}
\drawedge[dash={1.5}{.5},curvedepth=-1](A7,A8){}
\drawedge[dash={1.5}{.5},curvedepth=-1](A6,A7){}
\drawedge[dash={1.5}{.5},curvedepth=1](A5,A6){}
\drawedge[dash={1.5}{.5},curvedepth=1](A4,A5){}
\drawedge[dash={1.5}{.5},curvedepth=1](A3,A4){}
\drawedge[dash={1.5}{.5},curvedepth=1](A2,A3){}
\drawedge[dash={1.5}{.5},curvedepth=1](A1,A2){}
\drawedge[dash={1.5}{.5},curvedepth=1](A0,A1){}
\end{picture}
\caption{The DFA $\mathrsfs{E}_{12}$; solid and dashed arrows show the action of $a$ and $b$, respectively}\label{fig:e12}
\end{center}
\end{figure}

\begin{examplenew}
\label{examp:casas}
The automaton $\mathrsfs{E}_{12}$ is completely reachable, but the graph $\Gamma_1(\mE_{12})$ is not strongly connected.
\end{examplenew}

\begin{proof}
First, compute the graph $\Gamma_1(\mE_{12})$. The letter $a$ forces the edge $0\to 10$ and for each $k=1,2,\dots,11$, the word $ab^k$ forces the edge $k\to 10+k\!\pmod{12}$. These 12 edges form two cycles:
$0\to 10\to 8\to 6\to 4\to 2\to 0$ and $1\to 11\to 9\to 7\to 5\to 3\to 1$. It can be easily calculated that each word $w$ of defect 1 with respect to $\mathrsfs{E}_{12}$ is of the form $b^iub^j$, where $i,j\in\{0,1,\dots,11\}$ and $u$ is an arbitrary non-empty word in $\{a,ab^6a\}^*$. Using this and Lemma~\ref{lem:folklore}, once sees that the parity of $\excl(w)$ is the same as that of $\dupl(w)$ for every such $w$.  Therefore, $\Gamma_1(\mE_{12})$ contains two clusters $C_{\mathrm{even}}:=\{0,2,4,6,8,10\}$ and $C_{\mathrm{odd}}:=\{1,3,5,7,9,11\}$, and hence, $\Gamma_1(\mE_{12})$ is not \scn.

The clusters $C_{\mathrm{even}}$ and $C_{\mathrm{odd}}$ constitute the vertices of the graph $\Gamma_2(\mE_{12})$. The word $ab^{10}a$ forces the edge $C_{\mathrm{even}}\to C_{\mathrm{odd}}$ in $\Gamma_2(\mE_{12})$ because $\excl(ab^{10}a)=\{0,6\}\subset C_{\mathrm{even}}$ and $\dupl(ab^{10}a)=\{3,10\}$ shares a state with $C_{\mathrm{odd}}$.  The word $ab^{10}ab$ forces the opposite edge $C_{\mathrm{odd}}\to C_{\mathrm{even}}$ since $\excl(ab^{10}ab)=\{1,7\}\subset C_{\mathrm{odd}}$ while $\dupl(ab^{10}ab)=\{4,11\}$ shares a state with $C_{\mathrm{even}}$. Therefore, the graph $\Gamma_2(\mE_{12})$ is \scn. Now Theorem~\ref{thm:sufficient3} ensures that the automaton $\mathrsfs{E}_{12}$ is completely reachable.
\end{proof}

\section{Reset Threshold of Completely Reachable Automata}
\label{sec:resetthreshold}

We start with an observation that was mentioned in~\cite[Section 5]{BondarVolkov16} without proof. As the reader will see, it immediately follows from a combination of a few known facts.

\begin{proposition}
\label{prop:binary}
The \v{C}ern\'y conjecture holds for \cra\ with two input letters, that is, if a \cran\  $\mathrsfs{A}=\langle Q,\{a,b\}\rangle$ has $n$ states, then the reset threshold of $\mA$ does not exceed $(n-1)^2$, and the bound is tight.
\end{proposition}

\begin{proof}
If $n=1$, the claim is trivial so we assume that $n>1$. Every subset of the form $Q\dt w$, where $w$ is a non-empty word over $\{a,b\}$, is contained in either $Q\dt a$ or $Q\dt b$. At least one of the letters must have defect~1 since no subset of size $n-1$ is reachable otherwise. If the other letter has defect greater than 1, only one subset of size $n-1$ is reachable. Hence, one of the letters has defect~1 while the other has defect at most 1. For certainty, let $a$ stand for the letter of defect~1. If $b$ also has defect~1, then at most two subsets of size $n-1$ are reachable (namely, $Q\dt a$ and $Q\dt b$), and $\mathrsfs{A}$ can only be completely reachable provided that $n=2$. The automaton $\mathrsfs{A}$ is then nothing but the classical flip-flop, see Fig.~\ref{fig:flip-flop}.
\begin{figure}[bht]
\begin{center}
  \unitlength=3pt
\begin{picture}(25,12)(0,-7)
    \gasset{curvedepth=4}
    \node(A1)(0,0){$0$}
    \drawloop[loopangle=180](A1){$a$}
    \node(A2)(25,0){$1$}
    \drawloop[loopangle=0](A2){$b$}
    \drawedge(A1,A2){$b$}
    \drawedge(A2,A1){$a$}
\end{picture}
\caption{The filp-flop}\label{fig:flip-flop}
\end{center}
\end{figure}
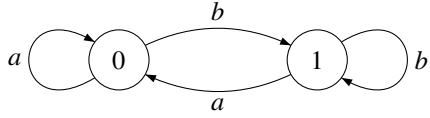

Obviously, the reset threshold of the flip-flop is 1 which equals $(2-1)^2$. Hence, we may assume that $b$ acts as a permutation of $Q$. Here we invoke the following fact.
\begin{lemma}
\label{lem:cyclic}
If $\mathrsfs{A}=\langle Q,\{a,b\}\rangle$ is a \cran\ in which the letter $b$ acts as a permutation of $Q$, then $b$ acts as a cyclic permutation.
\end{lemma}
Dubuc~\cite{Dubuc98} proved the \v{C}ern\'y conjecture for \sa\ in which one of the input letters acts a cyclic permutation. Combining this result with Lemma~\ref{lem:cyclic} yields the desired conclusion.
Example~\ref{examp:cerny} shows that the bound $(n-1)^2$ is tight.
\end{proof}

As for Lemma~\ref{lem:cyclic}, it was first stated in~\cite{BondarVolkov16} without proof. Then a slightly more general statement appeared in~\cite[Corollary 4]{Hoffmann21a}, again without proof. A proof of yet another generalization can be found in Appendix A of the arXiv version of~\cite{Hoffmann21a}. For the reader's convenience, we include a direct and self-contained proof of Lemma~\ref{lem:cyclic} here.

\begin{proof}[Proof of Lemma~\ref{lem:cyclic}]
Arguing by contradiction, suppose that the cyclic decomposition of the permutation induced by $b$ involves $k\ge 2$ independent cycles $\pi_1,\dots,\pi_k$. Let $Q_j$, $j=1,\dots,k$, stand for the subset of states moved by the cycle $\pi_j$. The letter $a$ is such that $Q\dt a\ne Q=\cup_{j=1}^kQ_j$ whence some $Q_j$ is not contained in $Q\dt a$. We fix a subset $Q_i\nsubseteq Q\dt a$. As the DFA $\mA$ is completely reachable, $Q_i$ is the image of a word in $\{a,b\}^*$. Let $w$ be a word of minimum length with $Q_i=Q\dt w$; observe that $|w|>0$ as $Q_i\ne Q$. Now we look at the rightmost letter of the word $w$. First assume that this letter is $b$, that is, $w=w'b$ for some $w'$. If $m$ is the least common multiple of the lengths of the cycles $\pi_1,\dots,\pi_k$, then the word $b^m$ acts as the identity transformation. Applying the word $b^{m-1}$ to the equality $Q\dt w'b=Q\dt w=Q_i$, we get $Q\dt w'=Q\dt wb^{m-1}=Q_i\dt b^{m-1}$. The action of $b$ restricted to $Q_i$ is the cyclic permutation $\pi_i$ whence $Q_i\dt b^{m-1}=Q_i$, and therefore, $Q\dt w'=Q_i$. As $|w'|<|w|$, this contradicts our choice of the word $w$. Thus, the rightmost letter of $w$ is $a$ whence $Q\dt w$ is contained in $Q\dt a$. Now the equality $Q_i=Q\dt w$ contradicts the choice of $Q_i\nsubseteq Q\dt a$.
\end{proof}

Now we turn to general \cra. We use an idea that comes from \cite{Trahtman:2011}. Let $\mathrsfs{A}=\langle Q,\Sigma\rangle$ be a DFA. Call a state $q\in Q$ \emph{avoidable} in $\mA$ if there exists a word $v\in\Sigma^*$ such that $q\notin Q\dt v$; the word $v$ is then said to \emph{avoid} $q$. In~\cite{Trahtman:2011}, avoiding words were used to construct a `halving' word for every \scn\ \san\ $\mA=\langle Q,\Sigma\rangle$, that is, a word over $\Sigma$ whose image size is at most $\frac12|Q|$.

Clearly, if a word $w$ resets a \san\ $\mA$ to a state $s$, then $w$ avoids all states of $\mA$ except $s$. If $\mA$ is strongly connected, there is a letter $a$ such that $s\dt a\ne s$, and therefore, the word $wa$ avoids $s$. Hence, in a \scn\ \san, every state is avoidable. It was claimed in \cite[Lemma~3]{Trahtman:2011} that in a \scn\ \san\ with $n$ states, each state is avoided by a word of length at most~$n$. This claim is wrong; see~\cite{GonzeJT:2015} for a counterexample. However, if restricted to \cra, the claim holds.

\begin{lemma}
\label{lem:avoiding}
In every \cran\ with $n$ states, each state is avoided by a word of length at most $n$.
\end{lemma}

\begin{proof}
Let $\mA=\langle Q,\Sigma\rangle$ be a \cran\ with $|Q|=n$ and $q\in Q$. Since the subset $Q{\setminus}\{q\}$ is reachable, there is a word $v\in\Sigma^*$ such that $Q\dt v=Q{\setminus}\{q\}$. Clearly, any such word $v$ avoids $q$. Now let $w=a_1a_2\cdots a_\ell$ with $a_1,a_2,\dots,a_\ell\in\Sigma$ be a word of minimum length satisfying $Q\dt w=Q{\setminus}\{q\}$. Let $Q_0:=Q$ and for each $i=1,\dots,\ell$, let $Q_:=Q_{i-1}\dt a_i$ so that $Q_\ell=Q{\setminus}\{q\}$. Observe that no sets in the sequence $Q_0,Q_1,\dots,Q_\ell$ can coincide. Indeed, if $Q_j=Q_k$ for some $0\le j<k\le\ell$, then $Q_0\dt w'=Q_\ell$ where the word $w'$ is obtained by cutting the non-empty factor $a_{j+1}\cdots a_{k-1}$ out of $w$, and this contradicts the choice of $w$. Thus, $n=|Q|>|Q_1|\ge|Q_2|\ge\cdots\ge|Q_\ell|=n-1$. Therefore, the sets $Q_1,\dots,Q_\ell$ are $\ell$ distinct $(n-1)$-element subsets of the $n$-element set $Q$ whence $\ell\le n$.
\end{proof}

For any \san\ $\mA=\langle Q,\Sigma\rangle$, consider the following procedure, which is a simplified version of (the correct part of) arguments in \cite{Trahtman:2011}.
\begin{algorithm}
  \setlength{\commentspace}{6cm}
  \begin{algorithmic}[1]
    \FUNC{Halving$(\mathrsfs{A})$}
    \STATE{take a letter $a\in\Sigma$ of maximum defect}
    \STATE\algcomment{0}{Initializing the current word}$w\gets a$
    \STATE\algcomment{0}{Initializing the current set}$P\gets Q\dt a$
    \WHILE{$|P|>\frac12|Q|$}
    \STATE{}take a state $q\in Q$  such that $q\dt w\notin\dupl(w)$
    \STATE{}take a word $u\in\Sigma^*$ of minimum length that avoids $q$
    \STATE\algcomment{1}{Updating the current word}$w\gets uw$
    \STATE\algcomment{1}{Updating the current set}$P\gets Q\dt w$
    \ENDWHILE
    \RETURN $w$
  \end{algorithmic}
\caption{Computing a word whose image size is at most $\frac12|Q|$}\label{alg:Avoiding}
\end{algorithm}

\begin{proposition}
\label{prop:halving}
For any \cran\ $\mA=\langle Q,\Sigma\rangle$ with $n$ states, Algorithm~\ref{alg:Avoiding} stops after at most $\lceil\frac{n}2\rceil-1$ repetitions of the main loop \textup(lines \textup{4--18} of the pseudocode\textup) and returns a word of length at most $n(\lceil\frac{n}2\rceil-1)+1$ with image size at most $\frac{n}2$.
\end{proposition}

\begin{proof}
Observe that $P=Q\dt w$ in the course of Algorithm~\ref{alg:Avoiding}. By the definition of the set $\dupl(w)$, we have $2|\dupl(w)|\le|Q|$. Therefore, until $|P|>\frac12|Q|$, we have $|P|>|\dupl(w)|$ and the set difference  $P{\setminus}\dupl(w)$ is not empty. This guarantees the existence of the state $q$ in line 5 of the pseudocode: for any state $p\in P{\setminus}\dupl(w)$, the unique state in $pw^{-1}$ suits the role of $q$. Now, if $u$ avoids $q$, then $uw$ avoids $p$, whence $Q\dt uw\subsetneqq Q\dt w$. We conclude that the size of the current set $P$ drops by 1 after every repetition of the main loop. The loop starts with $P=Q\dt a$ so that $|P|\le n-1$ and ends with $|P|\le\lfloor\frac{n}2\rfloor$. Hence the main loops repeats at most $(n-1)-\lfloor\frac{n}2\rfloor=\lceil\frac{n}2\rceil-1$ times. Lemma~\ref{lem:avoiding} ensures that a prefix of length at most $n$ is added to the current word $w$ at each repetition. Since $|w|=1$ when the main loop starts, the length of $w$ at the end of the loop does not exceed $n(\lceil\frac{n}2\rceil-1)+1$.
\end{proof}

We keep considering a fixed but arbitrary \cran\ $\mA=\langle Q,\Sigma\rangle$ with $n$ states. Proposition~\ref{prop:halving} provides a relatively short word $w_{1/2}$ with $|Q\dt w_{1/2}|\le\frac{n}2$. Now we proceed as in~\cite{Trahtman:2011}. Set $P_0:=Q\dt w_{1/2}$ and if $|P_0|>1$, let $u_1\in\Sigma^*$ be a word of minimum length with $|P_0\dt u_1|<|P_0|$. Then set $P_1:=P_0\dt u_1$ and if $|P_1|>1$, let $u_2\in\Sigma^*$ be a word of minimum length with $|P_1\dt u_2|<|P_1|$. We continue this process, getting the sequence $P_0,P_1,\dots$ of sets of decreasing size until we reach a singleton set. In the process, we also get the sequence $u_1,u_2,\dots$ of words,
and the word $w:=w_{1/2}u_1u_2\cdots$ is a reset word for $\mA$. In order to estimate the length of $w$, we invoke the key lemma from the proof of the Pin--Frankl bound on the \rl\ of \sa.

\begin{lemma}
\label{lem:pinfrankl}
Let $\mathrsfs{A}=\langle Q,\Sigma\rangle$ be a DFA with $n$ states, $P$ a $k$-element subset of $Q$ with $k> 1$, and $v\in\Sigma^*$ a word of minimum length with $|P\dt v|<|P|$. Then $|v|\le\binom{n-k+2}2$.
\end{lemma}

Lemma~\ref{lem:pinfrankl} gives the bounds $|u_1|\le\binom{n-\lfloor\frac{n}2\rfloor+2}2$, $|u_2|\le\binom{n-\lfloor\frac{n}2\rfloor+3}2$, and so on. Summing up all these inequalities, we see that the length of the product $u_1u_2\cdots$ is upper bounded by $\sum_{k=2}^{\lfloor\frac{n}2\rfloor}\binom{n-k+2}2$. This sum can be represented as
\[
\sum_{k=2}^{\lfloor\frac{n}2\rfloor}\binom{n-k+2}2=\sum_{k=2}^n\binom{n-k+2}2-\sum_{k=\lfloor\frac{n}2\rfloor+1}^{n}\binom{n-k+2}2,
\]
and the two sums in the right-hand side can be easily computed using \cite[formula (5.10)]{GKP06}. Namely, $\sum_{k=2}^n\binom{n-k+2}2=\binom{n+1}3$ and $\sum_{k=\lfloor\frac{n}2\rfloor+1}^{n}\binom{n-k+2}2=\binom{\lceil\frac{n}2\rceil+2}3$. Now elementary calculations give
\[
\binom{n+1}3-\binom{\lceil\frac{n}2\rceil+2}3=\begin{cases}
\dfrac{7n^3-6n^2-16n}{48}&\text{ for even $n$},\\[2ex]
\dfrac{7n^3-9n^2-31n-15}{48}&\text{ for odd $n$}.
\end{cases}
\]
Adding the estimate $|w_{1/2}|\le n(\lceil\frac{n}2\rceil-1)+1$ from Proposition~\ref{prop:halving}, we arrive to the main result of the section.
\begin{theorem}
\label{thm:cubic}
Every \cran\ with $n$ states has a reset word of length at most $\dfrac{7n^3+18n^2-64n+48}{48}$ if $n$ is even and $\dfrac{7n^3+15n^2-55n+33}{48}$ if $n$ is odd.
\end{theorem}

The upper bound from Theorem~\ref{thm:cubic} is still cubic in $n$, but its leading coefficient $\frac7{48}=0.14583\dots$ is a bit smaller than the leading coefficient $0.1654\dots$ of the best so far upper bound for the \rl\ of general \sa\ with $n$ states from~\cite{Shitov:2019}.

\begin{remarknew}
The above proof of Theorem~\ref{thm:cubic} shows that its result persists for \sa\ $\mathrsfs{A}=\langle Q,\Sigma\rangle$ with all $(|Q|-1)$-element subsets reachable. It can also be verified that, similarly to the result of Proposition~\ref{prop:binary}, the \v{C}ern\'y conjecture holds for \sa\ $\mathrsfs{A}=\langle Q,\{a,b\}\rangle$ with all $(|Q|-1)$-element subsets reachable.
\end{remarknew}

\section{Further work}
\label{sec:final}

The results of the present paper suggest several directions for further research. Here we briefly outline two such directions; many further open problem about \cra\ can be found in the final sections of \cite{BondarVolkov16,BondarVolkov18}.

\subsection{Reconstructing \cra\ from graphs and trees}
\label{subsec:graphandtrees}

In Section~\ref{sec:characterization} we assigned to any \cran\ $\mathrsfs{A}=\langle Q,\Sigma\rangle$ a sequence of graphs $\Gamma_1(\mA)$, $\Gamma_2(\mA)$, \dots, $\Gamma_k(\mA)$ such that the final graph in the sequence is \scn. The vertex set $Q_1$ of the graph $\Gamma_1(\mA)$ is $Q$, while for each $i=2,\dots,k$, the vertex set $Q_i$ of the graph $\Gamma_i(\mA)$ is the set of clusters of the preceding graph $\Gamma_{i-1}(\mA)$. The forest of clusters $\mathcal{F}_k(\mA)$ has the set $Q_1\cup Q_2\cup\dots\cup Q_{k-1}\cup Q_k$ as the vertex set and the relation `to be an element of' as the child-parent relation. It is convenient to add to $\mathcal{F}_k(\mA)$ the set $Q_{k+1}$ consisting of the unique cluster of the graph $\Gamma_k(\mA)$; this way we convert $\mathcal{F}_k(\mA)$ into a tree having $Q_{k+1}$ as the root. We denote this tree by $\mathcal{T}(\mA)$.

The graph sequence $\Gamma_1(\mA)$, $\Gamma_2(\mA)$, \dots, $\Gamma_k(\mA)$ provides a sort of stratification of the action of words in the DFA $\mA$ with respect to their defect: recall that the `new' edges of $\Gamma_{i}(\mA)$ added to those of the condensation of $\Gamma_{i-1}(\mA)$ are forced by words of defect $i$. The tree $\mathcal{T}(\mA)$ consists of the layers $Q_1,Q_2,\dots,Q_k,Q_{k+1}$ and registers the information about inclusions between clusters in different levels of the stratification. To what extent do these data (the graph sequence and the tree) determine the automaton $\mA$?

Given a DFA $\mA=\langle Q,\Sigma\rangle$, its \emph{singular semigroup} $\Sing(\mathrsfs{A})$ is the set of all transformations of the set $Q$ induced by the words in $\Sigma^*$ that have positive defect with respect to $\mA$.  Being defined via the action of words of positive defect, the graph sequence and the tree of $\mA$ depend on the semigroup $\Sing(\mathrsfs{A})$ only. Therefore, the question raised in the preceding paragraph actually asks to what extent the graph sequence and the tree of a \cran\ determine its singular semigroup.

It is not too hard to exhibit \cra\ with identical graph sequences and trees but different singular semigroups. For instance, in the \v{C}ern\'y automaton $\mC_n$, the word $(ab)^jab^i$ forces the edge $i\to i+j+1\pmod{n}$ for each $i=0,1,\dots,n-1$ and $j=0,1,\dots,n-2$ in the graph $\Gamma_1(\mC_n)$ so that $\Gamma_1(\mC_n)$ has all $n(n-1)$ possible edges between its $n$ vertices and, in particular, is \scn. Thus, the graph sequence of $\mC_n$ reduces to just $\Gamma_1(\mC_n)$ and the tree $\mathcal{T}(\mC_n)$ consists of $n+1$ vertices $n$ of which are leaves. If we add to the automaton $\mC_n$ an extra letter that swaps 0 and 1 and fixes all other states, the resulting automaton $\mC'_n$ will have the same graph sequence and the same tree. On the other hand, it follows from \cite[Theorem 3.18]{McAlister:1998} that for each $n>3$, the semigroup $\Sing(\mathrsfs{C}_n)$ is properly contained in $\Sing(\mathrsfs{C}'_n)$. Because of this and similar examples, a natural concretization of the general question stated above may consist of looking for a construction that, given a pair (graph sequence, tree), builds a \cran\  that is compatible with these data and has a minimum possible singular semigroup.

We plan to address this concretization in a follow-up paper. A related partial result was announced in~\cite[Section 4]{BondarVolkov16}. In order to state it, observe that the size of the $\Sing(\mathrsfs{A})$ of a \cran\ with $n$ states is at least $2^n-2$ transformations because for each proper non-empty subset $P$ of the state set, the semigroup must contain a transformation whose image is $P$. A \cran\ $\mathrsfs{A}=\langle Q,\Sigma\rangle$ is said to be \emph{minimal} if it attains this lower bound, that is, $|\Sing(\mathrsfs{A})|=2^{|Q|}-2$. The results of~\cite[Section 4]{BondarVolkov16} amount to, first, a complete classification of the trees of minimal \cra\ as so-called \emph{respectful trees} and, second, a construction that given a respectful tree, produces a minimal \cran\ with this tree. Now we can reveal that the automaton $\mE_5$, our running example in Section~\ref{sec:characterization}, is exactly the DFA produced this way from the tree obtained by adding the root $\{\{\{1,2\},\{3\}\},\{\{4,5\}\}\}$ to the forest shown in~Fig.~\ref{fig:f3e5}. In particular, $\mE_5$ is an example of a minimal \cran.

\subsection{Quantitative aspects}

Our partial results in Section~\ref{sec:resetthreshold} are rather minuscule if compared with the ultimate goal to prove the \v{C}ern\'{y} conjecture or at least a quadratic upper bound on \rl\ for \cra. In order to progress towards this goal, we need to achieve a better understanding of quantitative aspects of the graphs $\Gamma_k(\mA)$, $k=1,2\dotsc$, associated with a given \cran\ $\mA$. The key issue here consists in obtaining strong enough upper bounds on the length of words forcing the edges of $\Gamma_k(\mA)$. For $\mA$ with $n$ states, the proof of Theorem~\ref{thm:algorithm} shows that the length of words from $\ov{W}_k$ does not exceed $n^{2k}$, but this rough bound is insufficient for our purposes.

As a first step, one can consider the case $k=1$. If for a DFA $\mA$ with $n$ states, the graph $\Gamma_1(\mA)$ is \scn\ and its edges can be forced by words of length at most $n$, the proof of~\cite[Theorem 1]{BondarVolkov16} shows that $\mA$ satisfies not only the \v{C}ern\'{y} conjecture, but also an apparently stronger conjecture by Don~\cite[Conjecture~18]{Don16} who conjectured that in a completely reachable automaton with $n$ states, every subset of size $m$ is the image of a word of length at most $n(n-m)$. However, no such strong bound on the length of words forcing the edges of $\Gamma_1(\mA)$ holds in general. (For instance, in $\Gamma_1(\mC_n)$, the edge $n{-}1\to n{-}2$ is forced by the word $(ab)^{n-2}ab^{n-1}$ of length $3n-2$ and by no shorter word.) On the other hand, so far all examples are compatible with the conjecture that for any DFA $\mA$ with $n$ states, the edges of the graph $\Gamma_1(\mA)$ can be forced by words of length $O(n)$. Then the proof of~\cite[Theorem 1]{BondarVolkov16} allows one extract a quadratic in $n$ upper bound on the \rl\ of a DFA $\mA$ with $n$ states such that the graph $\Gamma_1(\mA)$ is \scn.

A further resource for improvement is provided by the observation that the proof of~\cite[Theorem 1]{BondarVolkov16} carries over when the graph $\Gamma_1(\mA)$ is replaced by any of its strongly connected spanned subgraphs. (Recall that given a simple graph $\Gamma$, its \emph{spanned subgraph} is any graph obtained by keeping all vertices of $\Gamma$ while removing some of its edges.) Thus, in order to deduce a quadratic in $n$ upper bound on the \rl\ of a DFA $\mA$ with $n$ states, it is sufficient to find a collection of words of defect 1 and length $O(n)$ that force edges forming a strongly connected spanned subgraph of $\Gamma_1(\mA)$. To illustrate that passing to a subgraph may be advantageous, look again at the \v{C}ern\'y automaton $\mC_n$. In its graph $\Gamma_1(\mC_n)$,  the edges $0\to1\to2\to\cdots\to n{-}1\to 0$ forced by the words $a,ab,\dots,ab^{n-1}$ of length at most $n$ constitute a \scn\ spanned subgraph while, as observed above, some edges of $\Gamma_1(\mC_n)$ cannot be forced by words of length less than $3n-2$. The approach based on constructing a strongly connected spanned subgraph  was utilized in~\cite[Theorem 7]{Gonzeetal2019} for \cra\ whose letters acting as permutations do not preserve any partition of the state set. Much earlier and in a less explicit form, the same idea was used in \cite{rystsov2000estimation} for \sa\ in which every letter either acts as a permutation or fixes all states but one.

\paragraph*{\emph{Acknowledgements}.} The authors are extremely grateful to the anonymous referees of the conference papers \cite{BondarVolkov16,BondarVolkov18}. The present article incorporates several enhancements based on the referees' observations and suggestions. The authors also thank Marina Maslennikova and Stefan Hoffmann for stimulating discussions.

\newcommand{\noopsort}[1]{} \newcommand{\singleletter}[1]{#1}
  \newcommand{\etal}{et al.}

\end{document}